\documentclass[twocolumn]{autart}
\usepackage{amsmath,amssymb,amsfonts}
\usepackage{algorithmic}
\usepackage{graphicx}
\usepackage{textcomp}
\usepackage{xcolor}
\usepackage{tikz}
\usepackage[numbers,sort&compress]{natbib}
\usepackage{url} 
\usepackage{subcaption}
\usepackage{lipsum} 
\usepackage{float} 
\usepackage{indentfirst}
\usepackage{enumitem}
\usepackage{etoolbox}

\theoremstyle{plain} 
\newtheorem{theorem}{Theorem} 
\newtheorem{proposition}{Proposition} 
 
\newtheorem{corollary}{Corollary} 
\theoremstyle{definition} 
\newtheorem{definition}{Definition} 
\newtheorem{assumption}{Assumption} 
\theoremstyle{remark} \newtheorem{remark}{Remark} 
\newenvironment{proof}[1][Proof]{
 \par\noindent\textit{#1.}\ }{ \hfill$\square$\par}

\def\BibTeX{{\rm B\kern-.05em{\sc i\kern-.025em b}\kern-.08em
    T\kern-.1667em\lower.7ex\hbox{E}\kern-.125emX}}
\makeatletter

\newcommand{\thmsep}{\par\addvspace{4pt plus 2pt minus 2pt}}

\AtBeginEnvironment{theorem}{\thmsep}
\AfterEndEnvironment{theorem}{\thmsep}
\AtBeginEnvironment{assumption}{\thmsep}
\AfterEndEnvironment{assumption}{\thmsep}
\AtBeginEnvironment{remark}{\thmsep}
\AfterEndEnvironment{remark}{\thmsep}
\AtBeginEnvironment{definition}{\thmsep}
\AfterEndEnvironment{definition}{\thmsep}
\renewcommand\subsection{%
  \@startsection{subsection}{2}{0pt}%
    {6pt plus 2pt minus 2pt}
    {3pt plus 1pt minus 1pt}
    {\normalfont\normalsize\itshape}
}
\renewcommand\subsubsection{%
  \@startsection{subsubsection}{3}{0pt}%
    {4pt plus 2pt minus 2pt}
    {2pt plus 1pt minus 1pt}
    {\normalfont\normalsize\itshape}
}

\AtBeginDocument{%
  \setlength{\abovedisplayskip}{4pt plus 2pt minus 2pt}%
  \setlength{\belowdisplayskip}{4pt plus 2pt minus 2pt}%
  \setlength{\abovedisplayshortskip}{4pt plus 2pt minus 2pt}%
  \setlength{\belowdisplayshortskip}{4pt plus 2pt minus 2pt}%
}
\makeatother

\begin{document}
\begin{frontmatter}

\title{Safe Navigation with Zonotopic Tubes: An Elastic Tube-based MPC Framework} 

\author[MSU]{Niyousha Ghiasi}\ead{ghiasini@msu.edu},    
\author[MSU]{Bahare Kiumarsi}\ead{kiumarsi@msu.edu},     
\author[MSU1]{Hamidreza Modares}\ead{modaresh@msu.edu}    

\address[MSU]{Department of Electrical and Computer Engineering, Michigan State University, East Lansing, MI, USA}  
\address[MSU1]{Department of Mechanical Engineering, Michigan State University, East Lansing, MI, USA}

\begin{abstract}
This paper presents an elastic tube-based model predictive control (MPC) framework for unknown discrete-time linear systems subject to disturbances. Unlike most existing elastic tube-based MPC methods, we do not assume perfect knowledge of the system model or disturbance realizations bounds. Instead, a conservative zonotopic disturbance set is initialized and iteratively refined using data and prior knowledge: data are used to identify matrix zonotope model sets for the system dynamics, while prior physical knowledge is employed to discard models and disturbances inconsistent with known constraints. This process yields constrained matrix zonotopes representing disturbance realizations and dynamics that enable a principled fusion of offline information with limited online data, improving MPC feasibility and performance.
The proposed design leverages closed-loop system characterization to learn and refine control gains that maintain a small tube size. By separating open-loop model mismatch from closed-loop effects in the error dynamics, the method avoids dependence on the size of the state and input operating regions, thereby reducing conservatism. An adaptive co-design of the tube and ancillary feedback ensures $\lambda$-contractive zonotopic tubes, guaranteeing robust positive invariance, improved feasibility margins, and enhanced disturbance tolerance.
We establish recursive feasibility conditions and introduce a polyhedral Lyapunov candidate for the error tube, proving exponential stability of the closed-loop error dynamics under the adaptive tube-gain updates. Simulations demonstrate improved robustness, enlarged feasibility regions, and safe closed-loop performance using only a small amount of online data.
\end{abstract}
\end{frontmatter}



\section{Introduction}

Safety-critical systems in aerospace, automotive, power, and robotics must satisfy hard state and input constraints while achieving high performance. For this purpose, model predictive control (MPC) has emerged as a leading methodology, optimizing system evolution over a receding horizon under explicit state and input constraints \cite{MPCbook1,bemporad1999robustMPCsurvey}. In practice, disturbances and model uncertainties prevent the plant from following the nominal model. Even small exogenous inputs or model errors can cause deviations between predicted and actual trajectories, jeopardizing constraint satisfaction. This motivates data-driven robust MPC, which leverages data to reduce model uncertainty while explicitly accounting for irreducible uncertainty to ensure feasibility and safety under realistic conditions \cite{MPCbook2,auto2}.

A prominent data-driven MPC approach is data-enabled predictive control (DeePC) 
\cite{coulson2019data,coulson2021robust,cummins2024deepc_hunt}, 
which computes finite-horizon optimal inputs directly from measured input–output trajectories. Although robust variants introduce regularization and slack variables to mitigate measurement noise and modeling errors, DeePC formulations typically lead to large-scale convex programs whose dimension grows with the size of the collected dataset and the prediction horizon. Moreover, since DeePC optimizes over trajectories rather than invariant sets, it does not provide explicit set-based guarantees such as robust positive invariance, nor does it yield a common Lyapunov certificate that ensures stability under time-varying operating conditions, changing data windows, or online updates. 

As a common robust MPC approach, min–max MPC optimizes the control policy against worst-case realizations, but this often results in conservative behavior and significant computational burden for online implementation \cite{minmax1,minmax2,minmax3}. Moreover, such methods typically assume known system models, which can further increase conservatism in uncertain environments. A more tractable class of robust methods is tube-based MPC (TMPC), which ensures robustness by confining the deviation between the actual and nominal trajectories to a robustly invariant tube, enabling constraint tightening and recursive feasibility under bounded disturbances \cite{rigid,AdjTube,AdjTube2,AdjTube3,crsec,knownf1,knownf2,un2,HTMPC1,un1}. In its classical form, TMPC decouples performance from robustness by combining a nominal MPC with a predesigned ancillary feedback and a fixed robust positively invariant (RPI) tube. Configuration-constrained tube MPC (CC-TMPC) \cite{CCTMPC,CCTMPC1} can be viewed as a special case of TMPC in which polytopic tubes and associated vertex control laws are optimized online within a fixed polyhedral template. While this added flexibility can reduce conservatism, CC-TMPC typically relies on worst-case, a priori uncertainty sets and fixed templates, which may remain conservative and lead to increased online computational burden as horizon length and template complexity grow, even with recent complexity reductions \cite{CCTMPC2}.

Despite TMPC’s computational efficiency, rigid-tube schemes typically assume known dynamics, use a fixed ancillary feedback, and rely on a predesigned RPI tube with uniform worst-case tightenings, leading to substantial conservatism \cite{rigid}. To reduce this, scaled-zonotope tightenings \cite{AdjTube}, elastic zonotopic tubes with horizon-wise scaling \cite{AdjTube2,AdjTube3}, and time-varying tube cross-sections via support-function tightenings \cite{crsec} have been proposed. Zonotope-based set-membership estimation has also been combined with TMPC \cite{knownf1}, and supervisory schemes filter unsafe inputs while ensuring constraint satisfaction \cite{knownf2}. Nevertheless, even with online tube scaling, these methods still assume known dynamics and fixed ancillary feedback, and the absence of data-driven model refinement maintains conservatism under real-world uncertainty. Recent work has extended TMPC to unknown dynamics by combining data-driven models with robust tubes. In \cite{un2}, a data-consistent model set is used with a fixed RPI tube, leading to conservative tube sizes. Multi-step predictors in \cite{HTMPC1} update a homothetic tube but rely on a fixed ancillary feedback. An operator-based propagation of zonotopic tubes is proposed in \cite{un1}, yet without adaptive scaling or feedback, conservatism remains and recursive feasibility is not guaranteed. Moreover, these methods lack mechanisms to incorporate prior system knowledge into their data-driven frameworks.

Unlike existing data-based TMPC methods that rely on open-loop learning, the proposed framework integrates closed-loop learning with prior physical knowledge, yielding a control-oriented and physically informed design. Prior knowledge has been combined with learning in stabilization settings~\cite{nar,ellips1}, but these works use the S-procedure to over-approximate ellipsoidal set intersections, introducing conservativeness and not extending to tube-based MPC. Here, we instead model tubes through closed-loop zonotopic error dynamics and iteratively refine them using data and physical constraints to discard inconsistent disturbances. This avoids set-intersection over-approximation and enables tractable, tighter uncertainty descriptions. Moreover, unlike elastic or homothetic TMPC with fixed ancillary feedback, we jointly adapt both the tube and feedback gains online, ensuring a compact, $\lambda$-contractive tube with reduced conservatism.

The main contributions of this work are summarized as follows: 
\begin{itemize}
  \item We construct a matrix zonotope capturing closed-loop model sets consistent with open-loop learned dynamics, then refine it into a constrained matrix zonotope by excluding models inconsistent with prior knowledge, enabling a principled, less conservative fusion of offline and limited online data.
  \item The error recursion separates open-loop model mismatch from closed-loop effects, enabling tighter tube propagation than lumping all uncertainty into a single set (as in \cite{un2}, which scales with both state and input sets and can be conservative).
  \item We alternate between solving TMPC for a nominal trajectory and updating both the tube and the ancillary feedback; the gain is not fixed a priori but updated at each iteration, reducing worst-case conservatism relative to fixed-gain designs.
  \item We design a $\lambda$-contractive zonotopic tube that maintains robust positive invariance while enforcing shrinkage of the error set over time, which increases feasibility margins and enhances disturbance tolerance.
  \item We derive verifiable conditions that guarantee recursive feasibility and propose a polyhedral Lyapunov candidate for the error tube that certifies exponential stability of the closed-loop error dynamics under the adaptive tube–gain updates.
\end{itemize}
\vspace{3pt}
\noindent \textbf{Notations and Definitions.} Real vectors and matrices are denoted by $\mathbb{R}^n$ and $\mathbb{R}^{n\times m}$, respectively.
$\mathbf{I}_p$ is the $p\times p$ identity matrix, and $\mathbf{0}_{n\times m}$ is the $n\times m$ zero matrix. $\mathbf  A\otimes \mathbf B$ denotes the Kronecker product of any matrices $\mathbf A$ and $\mathbf  B$. For $x\in\mathbb{R}^n$, $x^i$ denotes its $i$th component. For a matrix $X$, $X^\top$ is the transpose; $X^{i,:}$ and $X^{:,j}$ are its $i$th row and $j$th column, respectively; $X^{i,j}$ is the $(i,j)$ entry. $X^{\perp}$ denotes a matrix whose columns form the basis for the kernel of $X$; $X^\dagger$ is the Moore–Penrose pseudoinverse. We denote the column-wise vectorization of $X$ by $\text{vec}(X)$, obtained by stacking its columns into a single vector. For any vectors or matrices, $|\cdot|$ denotes the elementwise absolute value, and $\|\cdot\|$ denotes the infinity norm. Given $G=[\,G^{:,1}\ \cdots\ G^{:,v}\,]\in\mathbb{R}^{n\times m v}$ with blocks $G^{:,i}\in\mathbb{R}^{n\times m}$ and $Q \in\mathbb{R}^{m\times p}$, we define $
G\circ Q \;:=\; G\,(I_{v}\otimes Q)\in\mathbb{R}^{n\times pv}$. For any sets $\mathcal A,\mathcal B \subseteq \mathbb{R}^n$, $\mathcal A \oplus \mathcal B$ denotes the Minkowski sum and $\mathcal A \ominus \mathcal B$ denotes the (Pontryagin) Minkowski difference. We denote the set of vertices of a polytope $\mathcal{P}$ by $\mathrm{vert}(\mathcal{P})$.

\begin{definition}[(Constrained) zonotope {\cite{ID3}}]\label{defcz}
Let \(G\in\mathbb{R}^{n\times s}\) and \(c\in\mathbb{R}^n\).  
The constrained zonotope generated by \((G,c)\) with equality pair \((\bar A_c,\bar b_c)\),
\(\bar A_c\in\mathbb{R}^{n_c\times s}\), \(\bar b_c\in\mathbb{R}^{n_c}\), is
\begin{align*}
\mathcal Z &=  \langle G,c,\bar A_c,\bar b_c\rangle\\&
  = \Big\{x \in \mathbb{R}^n: x= c + G\bar\zeta \;\Big|\; \ \bar A_c\bar\zeta=\bar b_c,\, \big\Vert \bar\zeta  \big\Vert \leq 1 \Big\}.
\end{align*}

If no equalities are imposed (i.e., \(\bar A_c=\mathbf 0\), \(\bar b_c=\mathbf 0\)),
\(\mathcal Z\) reduces to the zonotope \(\langle G,c\rangle\).
\end{definition}

\begin{definition}[(Constrained) matrix zonotope {\cite{ID3}}]
Let \(G=[\,G^{:,1}\ \cdots\ G^{:,s}\,]\in\mathbb{R}^{n\times ms}\) with block columns
\(G^{:,i}\in\mathbb{R}^{n\times m}\), and let \(C\in\mathbb{R}^{n\times m}\).
With block-equality matrices \(\bar A_C=[\,\bar A_C^{:,1}\ \cdots\ \bar A_C^{:,s}\,]\in\mathbb{R}^{n_c\times m_cs}\) and
\(\bar B_C\in\mathbb{R}^{n_c\times m_c}\), the constrained matrix zonotope is
\begin{align*} 
\mathcal{M}\!=\!\big<G,C,\bar A_C,\bar B_C\big>\!=&\Big\{X \in \mathbb{R}^{n \times m}\!:\!X=C\!+\!\sum\limits_{i=1}^{s}G^{:,i}\bar\zeta^i \Big|\;\\& \sum\limits_{i=1}^{s} \bar A_C^{:,i}   \bar\zeta^i={\bar B_C}, \,\, \big\Vert \bar\zeta  \big\Vert \leq 1  \Big\}.
\end{align*}

In the absence of equalities (i.e., \(\bar A_C=\mathbf 0\), \(\bar B_C=\mathbf 0\)),
\(\mathcal M\) reduces to the matrix zonotope \(\langle G,C\rangle\).
\end{definition}

\section{Problem Formulation}
Consider the discrete-time linear system as 
\begin{equation}\label{system} 
x(t+1) = {A}^{\star}x(t) + B^{\star} u(t) + w(t),
\end{equation}
where $x(t) \in \mathcal{X} \subset \mathbb{R}^n$ is the system's state, $u(t) \in  \mathcal{U}\subset  \mathbb{R}^m$ is the control input, and $w(t) \in  \mathcal{Z}_w \subset \mathbb{R}^n$ is the additive disturbance. Moreover, $\mathcal{X}$, $\mathcal{U}$, and $\mathcal{Z}_w$ denote the admissible state, input, and disturbance sets, respectively.

\begin{assumption} \label{safeset}
The admissible state and input sets are given by the polytopes 
\begin{align*}
& \mathcal{X}=\mathcal{P} (H_x,h_x) = \\ & \big \{ x \in {\mathbb{R}^n}:H_x x  \le h_x, \,\, H_x \in \mathbb{R}^{q \times n}, \,\, h_x \in \mathbb{R}^q\big\},   
\end{align*}
and 
\begin{align*}
  &  \mathcal{U}=\mathcal{P} (H_u,h_u) = \\ & \big \{ u \in {\mathbb{R}^m}:H_u u  \le h_u , \,\, H_u \in \mathbb{R}^{v \times m}, \,\, h_u \in \mathbb{R}^v\big\},
\end{align*}
respectively.
\end{assumption}
\begin{assumption} \label{dist}
The disturbance set $\mathcal{Z}_w$ is a zonotope in $\mathbb{R}^{n}$ with $s_{w}$ generators. That is, $\mathcal{Z}_w=\big<G_h,c_h\big>$ for some $G_h \in \mathbb{R}^{n \times s_{w}}$, and $c_h \in  \mathbb{R}^{n}$.
\end{assumption}
\begin{assumption}\label{PK}
The actual parameter matrix \(\theta^\star=[A^\star\ \ B^\star]\in\mathbb{R}^{n\times(n+m)}\) is unknown but belongs to a given constrained matrix zonotope \(\mathcal{M}_{{prior}}=\langle G_\theta,\,C_\theta,\,A_\theta,\,B_\theta\rangle\),
with \(C_\theta\!\in\!\mathbb{R}^{n\times(n+m)}\),
\(G_\theta\!\in\!\mathbb{R}^{n\times (n+m)s_\theta}\),
\(A_\theta\!\in\!\mathbb{R}^{n_\theta\times m_\theta s_\theta}\),
and \(B_\theta\!\in\!\mathbb{R}^{n_\theta\times m_\theta}\).
\end{assumption}
\begin{assumption} \label{contr}
The pair \((A^\star,B^\star)\) is stabilizable within the prior set $\mathcal{M}_{prior}$.
\end{assumption} 

\begin{remark}
Assumption 1 is standard, as the admissible state and input sets are typically dictated by the the system’s physical limitations and environmental constraints. For example, actuator bounds, joint limits, workspace restrictions, and obstacle-avoidance requirements naturally define $\mathcal{X}$ and $\mathcal{U}$. 
Assumption 2 is also standard, as disturbances are typically bounded, while their exact realizations remain unknown. We therefore assume an initial conservative bound (e.g., a large zonotope) and subsequently refine it using available data and prior knowledge to accurately approximate the realized disturbances. Since the realized disturbances influence our proposed control framework—through the closed-loop representation—this refinement procedure substantially reduces conservatism. Assumption 3 is also standard, as coarse prior knowledge of the system bounds can typically be obtained from the physical principles governing the system dynamics. 
\end{remark} 

\begin{definition}[Robust invariant set (RIS) {\cite{SetB}}]
A set \(\mathcal P\subset\mathbb R^n\) is an RIS for \eqref{system} if, for every disturbance \(w(t)\in\mathcal Z_w\) and every initial condition \(x(0)\in\mathcal P\), it holds that \(x(t)\in\mathcal P\) for all \(t\ge 0\).
\end{definition} 

To learn an RIS, we leverage the notion of $\lambda$-contractive sets, introduced next, which guarantees that the resulting set is robustly invariant and that the state converges to the origin at a rate no slower than~$\lambda$.

\begin{definition}[$\lambda$–Contractive set {\cite{SetB}}]\label{contract}
For \(\lambda\in(0,1)\), a set \(\mathcal P\subset\mathbb R^n\) is \(\lambda\)–contractive for \eqref{system} if, for every \(t\ge 0\) and every \(w(t)\in\mathcal Z_w\), \(x(t)\in\mathcal P\) implies \(x(t{+}1)\in \lambda\,\mathcal P\).
This notion can be extended to a time–varying contraction factor $\lambda(t)\in(0,1)$ by requiring that the same set $\mathcal P$ satisfies, for every $t\ge 0$ and every $w(t)\in\mathcal Z_w$, that $x(t)\in\mathcal P$ implies $x(t{+}1)\in \lambda(t)\,\mathcal P$.
\end{definition}

We aim to ensure constraint satisfaction and robustness for the uncertain system~\eqref{system} by developing a TMPC scheme that leverages closed-loop learning and prior physical knowledge using only brief noisy data. A short input–state batch is used to construct uncertainty sets consistent with both data and priors, which are then mapped to their closed-loop counterparts. These sets drive a zonotopic error-tube initialization that supports iterative refinement of both the tube and the feedback gain.


\subsection{Prior-aware Modeling with Minimal Online Data}
To obtain a data-driven representation of the (open- and closed-loop) dynamics corresponding to system~\eqref{system}, we excite the system over a short time interval of length $T$ and collect input–state samples \cite{tesi}. The control inputs are organized as
\begin{align} \label{data-u}
U_0 := \begin{bmatrix} u(0) & u(1) & \cdots & u(T-1) \end{bmatrix} \in \mathbb{R}^{m \times T}.
\end{align}

The corresponding noisy state measurements are arranged as
\begin{align} \label{data-x0}
X_0 &:= \begin{bmatrix} x(0) & x(1) & \cdots & x(T-1) \end{bmatrix} \in \mathbb{R}^{n \times T},\\\label{data-x1}
X_1 &:= \begin{bmatrix} x(1) & x(2) & \cdots & x(T) \end{bmatrix} \in \mathbb{R}^{n \times T}. 
\end{align}

We define the stacked data matrix as
\begin{align*}
  D_0 := \begin{bmatrix}
X_0 \\ U_0
\end{bmatrix} \in \mathbb{R}^{(n+m) \times T}.
\end{align*}

The following assumption establishes the data-richness condition required for controller learning and is standard in data-driven control design \cite{tesi,ID3}. 

\begin{assumption}\label{assumption_5}
The data matrix $D_0$ has full row rank.
\end{assumption}

Let the unknown and unmeasurable disturbance sequence be
\begin{equation*}
W_0 :=  \begin{bmatrix} w(0) & w(1) & \ldots & w(T-1) \end{bmatrix} \in \mathbb{R}^{n \times T},
\end{equation*}
which is assumed to be unavailable for control design. The disturbance sequence is modeled by a matrix zonotope, formed via $T$-concatenation of the disturbance set $\mathcal{Z}_w$, given by \cite{ID3} 
\begin{align} \label{T-dis}
    \mathcal{M}_{\mathcal{Z}_w^T} := \langle G_w, C_w \rangle,
\end{align}
where $G_w = [\,G_w^{:,1},\,\ldots,\,G_w^{:,T s_w}\,]$ with each $G_w^{:,i} \in \mathbb{R}^{n\times T}$ is the generator matrix and $C_w \in \mathbb{R}^{n \times T}$ is the center matrix. It is assumed that $W_0 \in \mathcal{M}_{\mathcal{Z}_w^T}$.

\subsubsection{Open-loop Model Sets Consistent with Data and Prior Knowledge}
Using the data matrices defined in~\eqref{data-u}--\eqref{data-x1} and the system dynamics~\eqref{system}, we obtain the following relation
\begin{equation}\label{system-clN1} 
X_1 = A^{\star} X_0 + B^{\star} U_0 + W_0.
\end{equation}

Since the disturbance sequence satisfies $W_0\in\mathcal{M}_{\mathcal{Z}_w^T}=\langle G_w, C_w\rangle$, 
we write $W_0 = C_w + G_w \zeta_w$ with $\|\zeta_w\|\le 1$.  
Therefore,
\begin{align} \label{eqdata1}
   X_1 - C_w = \theta^\star D_0 + G_w \zeta_w, 
\end{align}
where $\theta^\star = [A^\star \quad B^\star]$. Right–multiplying \eqref{eqdata1} by the right pseudoinverse $D_0^\dagger$ yields
\[
\theta^\star = (X_1 - C_w)D_0^\dagger - (G_w \zeta_w)D_0^\dagger.
\]

Using the definition of the matrix zonotope and the Kronecker operator
$G_w\circ D_0^\dagger$, we obtain matrices $\theta = [A \quad B]$ (i.e., open-loop models) consistent with data as
\begin{align}
\mathcal{M}_{ol}
=
\left\langle
 -G_w\circ D_0^\dagger,\;
 (X_1 - C_w)D_0^\dagger
\right\rangle.
\label{olset}
\end{align}

However, not all disturbance realizations lead to a feasible solution for $\theta$. We therefore discard disturbances that are incompatible with the observed data. Specifically, the data equation \eqref{eqdata1}, which is $X_1 - C_w - G_w \zeta_w = \theta^\star D_0$, admits a solution $\theta$ only if the left-hand side lies in the column space of the right-hand side, i.e., the column space of $D_0$ \cite{ID5}. That is,
\[
\left(X_1 - C_w - G_w\zeta_w\right) D_0^\perp = 0,
\]
where $D_0^\perp$ spans the kernel of $D_0$.
Rearranging gives
\[
(X_1 - C_w) D_0^\perp
=
(G_w\zeta_w)\, D_0^\perp.
\]

Since $G_w\zeta_w = \sum_{i=1}^{T s_w} G_w^{:,i}\,\zeta_w^i$,
we finally obtain the data-consistency constraint
\begin{align}
(X_1 - C_w) D_0^\perp
=
\sum_{i=1}^{T s_w}
G_w^{:,i} D_0^\perp \,\zeta_w^i,
\qquad
\|\zeta_w\| \le 1.
\label{dcons}
\end{align}



By enforcing this constraint, the disturbance matrix zonotope in \eqref{T-dis} is refined to a constrained matrix zonotope, $\mathcal{M}_{w} = \langle G_w, C_w, A_w, B_w \rangle$,
where the constraint matrices \( A_w \) and \( B_w \) are defined as
\begin{subequations}\label{Aw}
\begin{align}
A_w &= G_w \circ D_0^{\bot} \in \mathbb{R}^{n \times (T-(m+n)) Ts_w}, \\
B_w &= (X_1 - C_w) D_0^{\bot} \in \mathbb{R}^{n \times (T-(m+n))}.
\end{align}
\end{subequations}

This refined disturbance set can then be used to further tighten the open-loop model set~\eqref{olset}. Before doing so, we also exclude disturbance realizations that would generate dynamics $\theta$ incompatible with the prior knowledge.
 From the data relation
\[
X_{1}=\theta\,D_{0}+W_{0},
\quad\Longrightarrow\quad
W_{0}=X_{1}-\theta D_{0},
\]
and by Assumption~\ref{PK} (i.e., $\theta\in\mathcal{M}_{prior}$), the
prior-consistent disturbance set is the matrix zonotope
\begin{equation*}
\mathcal{M}_{d}
\;=\;
\Big\langle
-G_{\theta}\!\circ\! D_{0},\;
X_{1}-C_{\theta}D_{0},\;
A_{\theta},\;
B_{\theta}
\Big\rangle,
\end{equation*}
with $G_{\theta}\!\circ\! D_{0}\in\mathbb{R}^{n\times T s_{\theta}}$ and
$X_{1}-C_{\theta}D_{0}\in\mathbb{R}^{n\times T}$. 

Consequently, the disturbances that are both data-realizable and prior-consistent are characterized by \(W_{0}\in\mathcal{M}_{dw}\), where \(\mathcal{M}_{dw}\) is obtained as the intersection $\mathcal{M}_{dw}=\mathcal{M}_{w}\cap\mathcal{M}_{d}$, which admits the constrained matrix-zonotope \cite{unifying}
\begin{align} \label{dw}
    \mathcal{M}_{dw}&= \Bigg< \!\Big[G_w  \!\quad\! \mathbf{0}_{n \times Ts_{\theta}} \Big], C_w ,\!\begin{bmatrix}
        \bar A_w  & \mathbf{0}  \\ 
        \mathbf{0} & \bar A_{\theta} \\
       \bar G_w & -\bar G_\theta
    \end{bmatrix}\!,\!\begin{bmatrix}
        \bar B_w  \\ 
       \bar B_{\theta} \\
       \bar C_\theta - \bar C_w 
    \end{bmatrix}\! \Bigg>
   \nonumber \\ & =\big<G_{dw},C_{dw},A_{dw},{B}_{dw} \big>,
\end{align}
where \(C_{dw}=C_w\) and \(\bar A_w,\bar B_w,\bar G_w,\bar A_\theta,\bar B_\theta,\bar G_\theta\) are zero-padded counterparts of \(A_w,B_w,G_w,A_\theta,B_\theta,G_\theta\) (Table~1, \cite{unifying}); when dimensions align, the barred matrices are identical to the originals, otherwise zero padding ensures consistency of the stacked blocks \(A_{dw}\) and \(B_{dw}\).

Accordingly, the open-loop system set in \eqref{olset}, refined to be consistent with both the data and the prior, is given by
\begin{align}\label{olset1}
{\mathcal{M}}^{c}_{{ol}} := \left\langle -G_{dw} \circ D_0^{\dagger}, \; (X_1 - C_{dw}) D_0^{\dagger}, \; A_{dw}, \; B_{dw} \right\rangle,
\end{align}
with $G_{dw}$, $C_{dw}$, $A_{dw}$ and $B_{dw}$ defined in \eqref{dw}.

\subsubsection{Closed-loop Parametrized Sets Consistent with Data and Prior Knowledge}

While the refined open-loop system set can be used to robustly satisfy control specifications, the models that best fit the data do not necessarily yield the best control performance. Control-oriented approaches instead directly characterize the closed-loop dynamics and learn a decision variable that parameterizes these dynamics to meet the desired specifications. We will show that, within our TMPC framework, this closed-loop characterization leads to significantly reduced conservatism. 

To construct closed-loop system sets, we parametrize the control gain via a time-varying matrix \( V_K(t) \in \mathbb{R}^{T \times n} \) such that \cite{tesi} 
\begin{subequations} \label{tesi1}
\begin{align}
    &K(t) = U_0 V_K(t),\\
    &X_0 V_K(t) = \mathbf{I}_n.
\end{align}
\end{subequations}

Then, multiplying both sides of \eqref{system-clN1} by $V_K(t)$ and using \eqref{tesi1}, the data-driven closed-loop system simplifies to
\begin{align*}
A^{\star} + B^{\star} K(t) = (X_1 - W_0) V_K(t).
\end{align*}

Since $W_0$ belongs to a set, the closed-loop dynamics associated with a given gain $K(t)$ (or $V_K(t)$) are also represented by a set. However, this set can be overly conservative if the disturbance set associated with disturbance realization $W_0$ is not refined, as is common in most existing methods. To reduce this conservatism, in this paper, we leverage the refined disturbance set represented in~\eqref{dw}. 

For a given matrix $V_K(t)$ and knowing that $W_0 \in \mathcal{M}_{dw}$, the set of time-varying closed-loop system matrices consistent with both the data and the prior knowledge at time $t$ is contained in the constrained matrix zonotope
\begin{align}\label{clset}
\mathcal{M}^{c}_{{cl}}(t) \!:=\! \left\langle -G_{dw}\! \circ\! V_K(t), (X_1-C_{dw}) V_K(t),A_{dw},B_{dw} \right\rangle,
\end{align}
with $G_{dw}$, $C_{dw}$, $A_{dw}$ and $B_{dw}$ defined in \eqref{dw}.  

\begin{remark}
The equality constraints imposed by $A_{dw}$ and $B_{dw}$ in the set description of \eqref{clset} define hyperplanes that intersect the original matrix zonotope of feasible disturbance realizations. This intersection refines the disturbance set and reduces its overall size. Consequently, the resulting closed-loop model set associated with a given control gain $K(t)$ \textup{(}parameterized via $V_K(t)$\textup{)} becomes tighter. This refinement significantly improves feasibility and performance, since the decision variable $V_K(t)$ is learned directly such that the resulting closed-loop model set satisfies the control specifications. In our TMPC framework, the closed-loop model set is repeatedly refined as new data become available, thereby significantly enlarging the region of attraction and improving closed-loop performance.
\end{remark}

\subsection{Problem Statement}


Given the uncertain system~\eqref{system} and the data–prior-consistent model sets $\mathcal{M}^{c}_{ol}$ and $\mathcal{M}^{c}_{cl}(t)$, our goal is to compute, at each time $t$, an MPC control sequence that guarantees $x(t)\in\mathcal{X}$ and $u(t)\in\mathcal{U}$ for all admissible models and disturbances. We further aim to reduce conservatism by shrinking uncertainty and error tubes over time using brief online data and prior knowledge. To achieve this, we develop an elastic, zonotopic TMPC framework that enables closed-loop learning, uncertainty refinement, and adaptive tube–gain design.

\section{Elastic Tube-based MPC Framework}
Before presenting our elastic TMPC scheme, we first highlight the complementary roles of the 
open-loop and closed-loop model sets constructed in Section~2. We used the open-loop set 
$\mathcal{M}^c_{{ol}}$ in \eqref{olset1} for separating the nominal model from the model mismatch in 
the error dynamics. However, relying solely on open-loop descriptions—common in most existing 
data-driven TMPC methods—forces the error bound to depend on the full admissible sets 
$\mathcal{X}$ and $\mathcal{U}$, which leads to large and conservative tubes, especially when the 
operating region is wide.  To overcome this limitation, we additionally construct a refined closed-loop model set $\mathcal{M}^c_{{cl}}(t)$ in \eqref{clset} that directly characterizes the closed-loop matrices $A^\star + B^\star K(t)$ consistent with the data and the priors. Using this closed-loop representation, the error dynamics depend only on the nominal quantities $(\bar{x}(t), \bar{u}(t))$, not on the entire state and input sets, and therefore yield significantly tighter error propagation. As we will show, leveraging both the open-loop model mismatch structure and the refined closed-loop dynamics enables a tube description that is substantially less conservative than approaches based solely on open-loop model sets, while simultaneously allowing the 
ancillary gain $V_K(t)$ to be learned in a control-oriented manner that contracts the tube size.

\subsection{Error and Tube Representation}
To design an elastic TMPC controller, we first decompose the uncertain system~\eqref{system} into a nominal system and an associated error dynamics. For the nominal part, we select the nominal system matrices \( \bar{A} \) and \( \bar{B} \) as the center of the open-loop matrix zonotope~\eqref{olset1}, i.e.,
\begin{align*}
\begin{bmatrix} \bar{A} & \bar{B} \end{bmatrix} := (X_1 - C_{dw}) D_0^{\dagger}.
\end{align*}

The nominal system evolves according to
\begin{equation*}
\bar{x}(t+1) = \bar{A} \bar{x}(t) + \bar{B} \bar{u}(t).
\end{equation*}

Thus, the system dynamics \eqref{system} can be redefined as
\begin{equation*}
x(t+1) = (\bar{A} + \Delta A) x(t) + (\bar{B} + \Delta B) u(t) + w(t),
\end{equation*}
where \(\Delta A = A^{\star} - \bar{A}\) and \(\Delta B = B^{\star} - \bar{B}\) represent model mismatches. Defining the control input of the actual system \eqref{system} as $u(t) = \bar{u}(t) + K(t) e(t)$, with the error state $e(t) := x(t) - \bar{x}(t)$, the error dynamics can be expressed using the closed-loop model set together with the open-loop model mismatch as
\begin{equation}\label{error-dyn}
e(t+1) = (A^{\star} + B^{\star} K(t)) e(t) + \begin{bmatrix} \Delta A & \Delta B \end{bmatrix} \begin{bmatrix} \bar{x}(t) \\ \bar{u}(t) \end{bmatrix} + w(t),
\end{equation}
where \( w(t) \in \mathcal{Z}_w \), and the actual closed-loop matrix satisfies
\(A^{\star} + B^{\star} K(t) \in \mathcal{M}^c_{cl}(t)\), with $\mathcal{M}^c_{cl}(t)$ defined in \eqref{clset}. Furthermore, based on \eqref{olset1}, the open-loop model mismatch pair satisfies
\begin{align*}
\begin{bmatrix} \Delta A & \Delta B \end{bmatrix} &\in \mathcal{M}^c_{ol} - \begin{bmatrix} \bar{A} & \bar{B} \end{bmatrix} \\&:= \left\langle -G_{dw} \circ D_0^{\dagger}, \mathbf{0}_{n\times (n+m)}, A_{dw}, B_{dw} \right\rangle.
\end{align*}

Hence, for given \( \bar{x}(t) \) and \( \bar{u}(t) \), the next error state \( e(t+1) \) belongs to the following constrained zonotope,
\begin{align}\label{error-zonotope}
e(t\!+\!1)& \!\in\!\!\Bigg< \!\!\begin{bmatrix} -G_{dw} \!\circ\! (V_K(t)e(t)) & -G_{dw} \!\circ\! (D_0^\dagger\! \!\begin{bmatrix} \bar{x}(t) \nonumber \\ \bar{u}(t) \end{bmatrix}) & G_h\!\end{bmatrix}\!,\\
&(X_1 - C_{dw}) V_K(t)\, e(t) + c_h,\text{vec}(A_C),\text{vec}(B_C) 
\!\Bigg>,
\end{align}
where 
\begin{align*} 
 &  A_C=\begin{bmatrix}
         A_{dw}  & \mathbf{0} & \mathbf{0} \\
         \mathbf{0} & A_{dw}  &  \mathbf{0}\\
          \mathbf{0} &  \mathbf{0} &  \mathbf{0}
    \end{bmatrix}, \quad  B_C=\begin{bmatrix}
         B_{dw}  \\ 
         B_{dw}\\
        \mathbf{0}
    \end{bmatrix}. 
\end{align*}

According to~\eqref{error-zonotope}, the error set at time $t$ is a bounded constrained zonotope
and therefore a polytope. Hence, it admits an equivalent half-space representation
\begin{equation*} 
\mathcal{E}(t)
=
\mathcal{P}\bigl(H_e(t),h_e(t)\bigr)
:=
\{\, e\in\mathbb{R}^n \;|\; H_e(t)\, e \le h_e(t) \,\},
\end{equation*}
for some matrix $H_e(t)\in \mathbb{R}^{q_e\times n}$ and vector $h_e(t)\in\mathbb{R}^{q_e}_{>0}$.


\begin{remark}
A primary source of conservatism in existing data-driven TMPC methods stems from deriving the error propagation exclusively under the open-loop dynamics, together with a fixed ancillary feedback gain $K$ selected a priori. Specifically, given a learned open-loop model pair $(A,B)$ and a fixed stabilizing gain $K$, the deviation evolves as
\begin{equation}\label{error-dyn2}
e(t+1) = (\bar A + \bar B K) e(t) + \begin{bmatrix} \Delta A & \Delta B \end{bmatrix} \begin{bmatrix} {x}(t) \\ {u}(t) \end{bmatrix} + w(t),
\end{equation}
which requires the tube to scale with the sizes of both $\mathcal{X}$ and $\mathcal{U}$. When these sets are large, the resulting error tube becomes excessively conservative, which is a typical drawback of open-loop–based schemes (see, e.g.,~\cite{un2}). In contrast, the proposed framework mitigates this inflation by explicitly modeling the closed-loop deviation dynamics. Leveraging the refined disturbance set $\mathcal{M}_{dw}$ in \eqref{dw} and the gain parametrization $K(t)=U_0V_K(t)$, the deviation evolves by \eqref{error-zonotope}, which depends solely on the \emph{nominal} values $(\bar{x}(t),\bar{u}(t))$, rather than the full admissible sets $\mathcal{X}$ and  $\mathcal{U}$. This eliminates the dependence of the error tube on the operating region size, producing significantly tighter tubes. Furthermore, the decision variable $V_K(t)$ is optimized to enforce contraction of the error tube. As a result, the closed-loop characterization not only avoids the conservatism inherent in open-loop error propagation but also enables systematic reduction of tube uncertainty, thereby enlarging the 
region of attraction and improving closed-loop performance.
\end{remark} 

\begin{assumption}\label{ass:tube-shape}
The tube uses the same facet directions as the admissible state set, i.e., $H_e(t) \equiv H_x =: H_e$ for all $t$. Consequently, the tube cross-section at time $t$ is 
\begin{equation} \label{eset}
\mathcal{E}(t)
:=
\mathcal{P}\bigl(H_e,h_e(t)\bigr)
=
\{\, e\in\mathbb{R}^n \mid H_e e \le h_e(t) \,\},
\end{equation}
where $h_e(t) \in \mathbb{R}^q_{>0}$ is a time-varying scaling vector.
\end{assumption}

\begin{remark}
Under Assumption~\ref{ass:tube-shape}, the tube is $\mathcal{E}(t)=\mathcal{P}(H_e,h_e(t))$, where the fixed facet matrix $H_e$ aligns the tube geometry with the admissible set $\mathcal{X}$ and $h_e(t)$ scales the tube over time. The vector $h_e(t)$ is chosen so that the constrained-zonotope error set~\eqref{error-zonotope} lies inside $\mathcal{E}(t)$, yielding a polyhedral outer approximation. This alignment simplifies TMPC: (i) state tightening reduces to updating $h_x-h_e(t)$ without recomputing facet orientations, and (ii) only $h_e(t)$ is updated online, substantially reducing decision variables and constraints. Although optimizing or enriching $H_e$ could reduce conservatism, it greatly increases complexity; fixing $H_e$ therefore provides an effective balance between accuracy and tractability.
\end{remark}



We now use the error dynamics \eqref{error-zonotope} to design an elastic TMPC. Instead of planning on the disturbance-affected state, TMPC optimizes a nominal trajectory and ensures robustness by keeping the actual trajectory within a time-varying error tube $\mathcal{E}(t)$. Safety is enforced by tightening the state and input constraints using the tube cross-sections, so feasibility of the nominal MPC guarantees feasibility of the true closed loop for all models in the learned sets. The main challenge is to design a tube $\mathcal{E}(t)$ and ancillary gain $V_K(t)$ that contract the error, reduce conservatism, and enlarge the feasible region, motivating the tube–gain update problem introduced next.

\subsection{Elastic and Adaptive Tube-based MPC}
Given the horizon length $N$, the decision variable $V_K(t)$, the error tube $\mathcal{E}(t)$, $\mathcal{X}$, $\mathcal{U}$, and terminal set $\mathcal{T}$, the TMPC framework can be formulated as
\begin{subequations} \label{TMPC}
\begin{align}
&\min_{\{\bar{u}(t|t),\dots,\bar{u}(t+N-1|t)\}}  V_\mathcal{T} + \sum_{k=0}^{N-1} \mathcal{L}({\bar x(k|t),\bar u(k|t)}), \\
&\text{s.t.}\quad  \forall k \in \{0, \dots, N-1\}:\nonumber\\
\quad & \bar{x}(t+k+1|t) = \bar{A} \bar{x}(t+k|t) + \bar{B} \bar{u}(t+k|t),\\\label{TMPCX}
& \bar{x}(t+k|t) \oplus \mathcal{E}(t) \subseteq \mathcal{X},  \\\label{TMPCU}
& \bar{u}(t+k|t) \oplus  U_0 V_K(t) \mathcal{E}(t) \subseteq \mathcal{U},  \\\label{TMPCT}
& \bar{x}(t+N|t) \in \mathcal{T},
\end{align}
\end{subequations}
where $V_{\mathcal T}$ is the terminal cost, and the stage cost is
\[
\mathcal{L}\big(\bar x(k|t),\bar u(k|t)\big)
:= \bar x(k|t)^\top Q\,\bar x(k|t)
 + \bar u(k|t)^\top R\,\bar u(k|t),
\]
with $Q = Q^\top \succ 0$ and $R = R^\top \succ 0$ denoting the given state
and input weighting matrices, respectively. 

\begin{assumption}\label{terminal}
A terminal set $\mathcal{T}$, a terminal control law $K_\mathcal{T}(\cdot)$, and a terminal cost $V_{\mathcal{T}}:\mathbb{R}^n\to\mathbb{R}_{\ge 0}$ are given such that:
\begin{enumerate}[label=(\roman*)]
\item $\mathcal{T}$ is nonempty, compact, contains the origin, and satisfies the tightened
state and input constraints in \eqref{TMPCX}--\eqref{TMPCU} at $k = N$ under $K_\mathcal{T}(\cdot)$.
\item $\mathcal{T}$ is positively invariant for the nominal closed-loop, i.e.,
      for all $\bar x \in \mathcal{T}$,
      $ \bar A \bar x + \bar B K_\mathcal{T}(\bar x) \in \mathcal{T}$.
\item $V_{\mathcal{T}}$ is positive definite on $\mathcal{T}$ and, for all $\bar x \in \mathcal{T}$,
$      V_{\mathcal{T}}(\bar A \bar x + \bar B K_\mathcal{T}(\bar x)) \;\le\;
      V_{\mathcal{T}}(\bar x) - \mathcal{L}\big(\bar x, K_\mathcal{T}(\bar x)\big).$
\end{enumerate}
\end{assumption}

\begin{assumption}\label{init}
The initial nominal state $\bar{x}(0|0)$ is selected such that $x(0) \in \bar{x}(0|0) \oplus \mathcal{E}(0)$.
\end{assumption}
\begin{remark}
Assumption~\ref{init} is a standard requirement in TMPC. It states that the actual initial state $x(0)$ lies in the tube obtained by shifting the error set tube $\mathcal{E}(0)$ around the nominal state $\bar{x}(0|0)$. This condition is the starting point to guarantee recursive feasibility and constraint satisfaction of TMPC.
\end{remark}


In the proposed TMPC framework~\eqref{TMPC}, the nominal control input $\bar{u}(t)$ is designed based on an error tube $\mathcal{E}(t)$ and an ancillary feedback gain $V_K(t)$. Unlike many TMPC schemes that use a fixed, offline-designed tube and gain, the proposed method iteratively contracts $\mathcal{E}(t)$ and updates
$V_K(t)$ online. That is, the proposed framework performs the following two-stage procedure: 
\noindent 1) \emph{Tube-based MPC step:} At every time step $t$, it solves the nominal MPC problem \eqref{TMPC} to compute the nominal state and 
input sequences $(\bar{x}(t+k|t),\bar{u}(t+k|t))$ for $k = 0, \cdots, N-1$, and apply to the plant the input
\begin{equation} \label{optu}
    u(t)=\bar{u}^{\star}(t|t)+U_{0}V_{K}(t)e(t),
\end{equation}
where $\bar{u}^\star(t\mid t)$ is the first element of the optimal input sequence $\bar{u}^\star(t)=[\bar{u}^\star(t\mid t),...,\bar{u}^\star(t+N-1\mid t)]$ of~\eqref{TMPC} at time $t$.

\noindent 2) \emph{Tube–gain update step:} Using the newly obtained trajectory after each MPC step, the framework updates the ancillary gain $V_K(t)$, which in turn refines the closed-loop model set $\mathcal{M}^{c}_{cl}(t)$ and determines a new tube scaling that enforces tube shrinkage. The updated $V_K(t)$ and refined tube are then applied to the MPC problem at the next time step.

For the second step, $V_K(t)$ is re-optimized using the following theorem to yield a $\lambda(t)$-contractive error tube. Fig.~\ref{diagram} illustrates this closed-loop interaction 
between nominal planning and tube–gain learning. 


\begin{theorem} 
\label{thm-LPf}
Consider the system~\eqref{system} with error dynamics~\eqref{error-zonotope} under
Assumptions~1--8, and let $\mathcal{E}(t)$ denote the error tube at time $t$ as defined in~\eqref{eset}.
Suppose that, at time $t$, the optimization problem
\begin{subequations}
\label{LPf}
\begin{align} 
  &\min_{P(t),\,V_K(t),\,\rho(t),\,\lambda(t)} \quad  \rho(t) + \sigma \lambda(t), \\
  &\quad\text{s.t.} \nonumber \\\label{21b}
  &\quad P(t)\, h_e(t) \le \lambda(t)\, h_e(t) - H_e c_h - \rho(t)\, l(t) - z(t) - y, \\
  &\quad P(t)\, H_e = H_e \bigl(X_1 - C_{dw}\bigr) V_K(t), \\\label{hupdate}
  &\quad h_e(t{+}1) = \lambda(t)\, h_e(t),\\
  &\quad \| V_K(t) \| \le \rho(t), \\
  &\quad X_0 V_K(t) = \mathbf{I}_n, \\
  &\quad P(t) \ge \mathbf{0}_{q \times q}, \\
  &\quad 0 < \lambda(t) < 1,
\end{align}
\end{subequations}
with $\sigma > 0$ admits a solution $(P(t), V_K(t), \rho(t), \lambda(t))$.
Then the tube $\mathcal{E}(t)$ is $\lambda(t)$-contractive
(Definition~\ref{contract}), and hence it is an RPI set for the error dynamics \eqref{error-zonotope}. Here $y=[y^1,\dots,y^q]^{\top}$, $l(t)=[l^1(t),\dots,l^q(t)]^{\top}$, and $z(t)=[z^1(t),\dots,z^q(t)]^{\top}$,
where $y^j$, $l^j(t)$, and $z^j(t)$ are defined as
\begin{align} \label{yj}
& y^j  = \sum\limits_{i=1}^{s_w} \big| {H^{j,:}_e}  {G^{:,i}_h} \big|, 
\end{align}
\begin{subequations}
\label{lj}
\begin{align}
  l^j(t)  = {M_e}(t)\max\limits_{\beta} \,  & \big\|\sum\limits_{i=1}^{s_c}   \beta^i H^{j,:}_{e} G^{:,i}_{dw}  \big\Vert,     \,   \\ {\rm{s}}{\rm{.t}}{\rm{.}}\quad & \sum\limits_{i=1}^{s_c} {A}^{:,i}_{dw}\beta^i={B}_{dw}, \\  
  & \big\Vert \beta  \big\Vert \leq 1, 
\end{align}
\end{subequations}
\begin{subequations}
\label{zj}
\begin{align}
  z^j(t)  = \max\limits_{\beta} \,  & \big\|\sum\limits_{i=1}^{s_c}   \beta^i {{H^{j,:}_e}G^{:,i}_{dw}} D_0^\dagger\begin{bmatrix} \bar{x}(t) \\ \bar{u}(t) \end{bmatrix}  \big\Vert,     \,   \\ {\rm{s}}{\rm{.t}}{\rm{.}}\quad & \sum\limits_{i=1}^{s_c} {A}^{:,i}_{dw}\beta^i={B}_{dw}, \\  
  & \big\Vert \beta  \big\Vert \leq 1, 
\end{align}
\end{subequations}
respectively, where $s_c = Ts_w+s_\theta$.
\end{theorem}
\begin{figure}[t]
    \centering
    \includegraphics[width=0.4\textwidth, keepaspectratio]{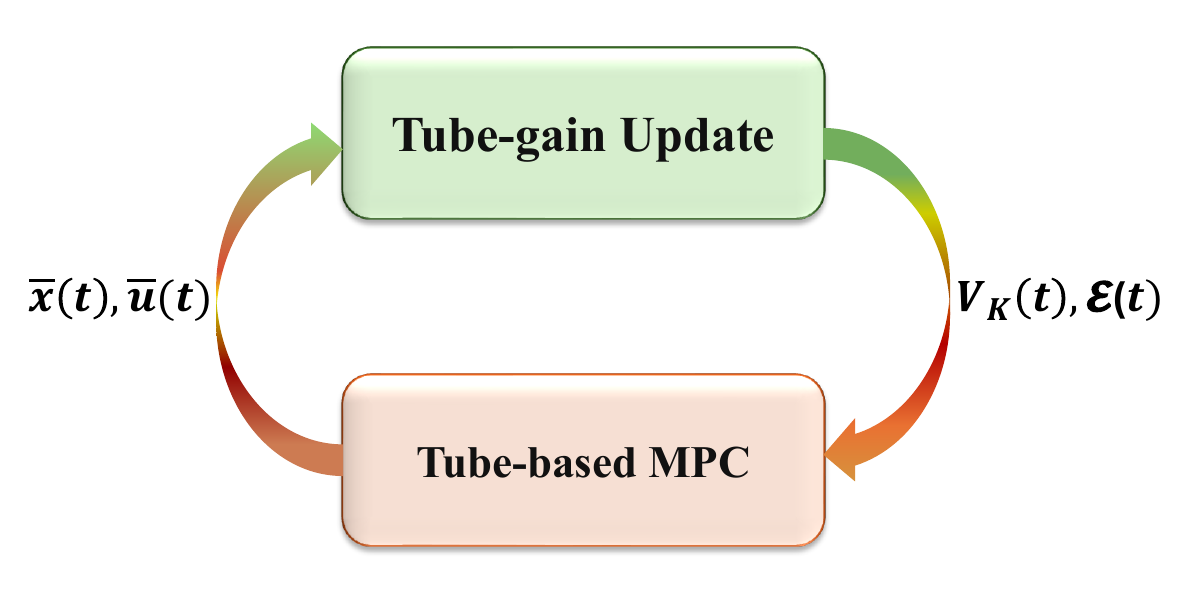} 
    \caption{\small{Alternating MPC and Tube–gain Update. MPC produces the nominal pair $(\bar{x}(t),\bar{u}(t))$, which is passed to the Tube–gain Update module. The module returns the contractivity factor $\lambda(t)$ and ancillary gain $V_K(t)$, defining the next error tube via $h_e(t{+}1)=\lambda(t)\,h_e(t)$; these are fed back to the Tube-based MPC module for the subsequent solve.}}
    \label{diagram}
\end{figure}
\begin{proof}
For the closed-loop error dynamics \eqref{error-dyn} and a contraction factor \(\lambda(t)\in(0,1)\), the polytope \(\mathcal{E}(t)\) in \eqref{eset} is $\lambda(t)$-contractive if $\gamma^j(t) \le \lambda(t) {h_e^j}(t)$, where $\gamma^j(t)=\max_{e(t)} \, H^{j,:}_{e} e(t+1)$ for $j=1,...,q$, whenever $H_e e(t) \le h_e(t)$. Using \eqref{error-zonotope}, define
\begin{align*} 
 & \bar \gamma^j(t)  = \max\limits_{e(t)}  \Big( {H^{j,:}_e} (X_1-C_{dw}) V_K(t) \, e(t)+{H^{j,:}_e} c_h+ \nonumber \\ &   \Big|\sum\limits_{i=1}^{s_c}   \beta^i {H^{j,:}_e} G^{:,i}_{dw} V_K(t) e(t) \Big| + \big|\sum\limits_{i=1}^{s_c} \beta^i{H^{j,:}_e} {G^{:,i}_{dw}} D_0^\dagger\begin{bmatrix} \bar{x}(t) \\ \bar{u}(t) \end{bmatrix}\big| +\nonumber \\& \sum\limits_{i=1}^{s_w} \big| {H^{j,:}_e}  {G^{:,i}_h} \big| \Big), \nonumber \\  &
 {\rm{s}}{\rm{.t}}{\rm{.}} \,\, \sum\limits_{i=1}^{2s_c+s_w} {A}^{:,i}_{C} \zeta^i={B}_{C},  \quad \big\Vert \zeta \big\Vert \leq 1, \nonumber \\  & \quad \quad H_e e(t) \le h_e(t),
\end{align*}
where $s_c = Ts_w+s_\theta$, and the last three terms of the cost function are obtained using the fact that for $\zeta=[\beta^{\top} \,\,\, \beta^{\top} \,\,\, \eta^{\top}]^\top \in \mathbb{R}^{2s_c+s_w}$, one has
\begin{align} \label{ineql}
   &  {H^{j,:}_e} \! \begin{bmatrix} -G_{dw} \circ (V_K(t)\,e(t)) & -G_{dw} \circ (D_0^\dagger \begin{bmatrix} \bar{x}(t) \\ \bar{u}(t) \end{bmatrix}) & G_h\end{bmatrix}\! \zeta  \!\le\!  \nonumber \\ & 
     \Big|\!\sum\limits_{i=1}^{s_c} \beta^i {H^{j,:}_e} G^{:,i}_{dw} V_K(t) e(t) \!\Big|   \!+\! \big|\!\sum\limits_{i=1}^{s_c} \beta^i {H^{j,:}_e} {G^{:,i}_{dw}} D_0^\dagger\begin{bmatrix} \bar{x}(t) \\ \bar{u}(t) \end{bmatrix}\!\big| \nonumber \\&+ \sum\limits_{i=1}^{s_w} \big| {H^{j,:}_e}  {G^{:,i}_h} \big|.
\end{align}

Based on \eqref{ineql}, we have \( \gamma^j(t) \le \bar{\gamma}^j(t) \), and thus, the polytope \(\mathcal{E}(t)\) in \eqref{eset} is $\lambda(t)$-contractive if \( \bar{\gamma}^j(t) \le \lambda(t) h^j_{e}(t) \) for \( j = 1, \dots, q \). To bound the term involving $\beta^i$, we exploit that $H_e e(t) \le h_e(t)$
implies $e(t)$ lies in a polytope, which admits a multidimensional interval
enclosure $e(t) \in [\underline{e}(t), \overline{e}(t)]$ with
$\underline{e}(t), \overline{e}(t) \in \mathbb{R}^n$ and
$\underline{e}(t) \le \overline{e}(t)$~\cite{Althoff}. From this, a bound on the norm of \( e(t) \) can be derived as \( \|e(t)\| \le M_e(t) \), where \( M_e(t) \) is a function of \( h_e(t) \) and \( H_e \), and hence a variable in the optimization. Therefore, the term involving \( \beta \) can be bounded as
\begin{align*}
&\left| \sum_{i=1}^{s_c} \beta^i {H^{j,:}_e} G^{:,i}_{dw} V_K(t) e(t) \right|
\le\\ &\quad \quad \quad \quad \quad \left\| \sum_{i=1}^{s_c} \beta^i {H^{j,:}_e} G^{:,i}_{dw} \right\| \|V_K(t)\| \|e(t)\|,
\end{align*}
where \( \|V_K(t)\| \le \rho(t) \) is enforced. Since \( M_e(t)\) depends on \( h_e(t)\), this product is bilinear in \( \rho(t) \) and \( h_e(t) \), hence nonconvex. To avoid bilinearity, we use \( e(t) = x(t) - \bar{x}(t) \) and initialize the error tube as \(\mathcal{E}(0):=\mathcal{X}\). Given the bound \(\|x(t)\|\le M_x\), it follows that \(\|e(0)\|\le M_e(0) = M_x\). This bound is used only to initialize the tube--gain update problem~\eqref{LPf}, where we optimize over the feedback gain $V_K$ and the scalar contraction factor $\lambda$ using this error bound; this initial tube is not passed to the TMPC optimization. In~\eqref{LPf}, the tube is refined by explicitly accounting for the closed-loop model sets and the open-loop mismatch in the error dynamics: the updated value of $\lambda$ determines a new tube scaling $h_e$ via~\eqref{hupdate}, and the corresponding error bound $M_e$ is updated accordingly. The refined error tube returned by~\eqref{LPf} is then used in the TMPC problem~\eqref{TMPC}. After solving the TMPC problem \eqref{TMPC} and computing the nominal trajectory \( \bar{x}(t) \) and \( \bar{u}(t) \), we pass this nominal solution to the tube--gain update problem~\eqref{LPf} to re-optimize the feedback gain $V_K(t)$ and the scalar contraction factor $\lambda(t)$, and to accordingly update the tube scaling $h_e(t)$ and the bound $M_e(t)$. Repeating this at each step removes bilinearity and progressively tightens the error bound.

Thus, the \( \bar{\gamma}^j(t) \) can be upper-bounded as $\bar \gamma^j(t) \le \hat \gamma^j(t)$, where 
\begin{align*} 
 & \hat \gamma^j(t)  = \max\limits_{e(t)}  \Big( {H^{j,:}_e} (X_1-C_{dw}) V_K(t) \, e(t)+{H^{j,:}_e} c_h+\\& \rho(t)\, l^j(t)+z^j(t)+y^j \Big), \nonumber \\  &
 {\rm{s}}{\rm{.t}}{\rm{.}} \,\, \quad H_e e(t) \le h_e(t),
\end{align*}
where $y^j$, $l^j(t)$, and $z^j(t)$ are defined in \eqref{yj}, \eqref{lj}, and \eqref{zj}, respectively, and represent the contribution of the disturbance zonotope, the effect of the
closed-loop model uncertainty, and the contribution of the open-loop
model mismatch as a function of $(\bar{x}(t),\bar{u}(t))$. Using duality, $\lambda(t)$-contactivity is satisfied if $\tilde \gamma^j(t) \le \lambda(t) {h_e^j}(t)$ where 
\begin{align*} 
&\tilde \gamma^j(t)  \!=\! \min\limits_{\alpha^j(t)} {\alpha^j(t)}\!^{\top}\! h_e(t)\!+\!  {H^{j,:}_e} c_h\!+\!  l^j(t)  \big\Vert {V_K(t)} \big\Vert\!+\! z^j(t)\!+\! y^j, \label{a1} \nonumber \\&
{\alpha^j(t)}^{\top} H_e =  {H^{j,:}_e} (X_1-C_{dw})  V_K(t),  \nonumber \\&
{\alpha^j(t)}^{\top} \ge \mathbf{0}_{1 \times q},  
\end{align*}
and $\alpha^j(t) \in \mathbb{R}^q$ is the dual variable associated with the 
$j$-th facet constraint in $H_e e(t) \le h_e(t)$. Collecting these dual variables 
for all facets $j=1,\dots,q$ into the matrix ${P(t)}=[\alpha^1(t),....,\alpha^q(t)]^{\top} \in \mathbb{R}^{q \times q}$, yields an elementwise nonnegative matrix $P(t)$, since each $\alpha^j(t)$ is 
nonnegative for all $j=1,\dots,q$. Stacking the facetwise conditions then yields the matrix inequalities
\begin{align*}
&P(t)\,h_e(t) \;\le\; \lambda(t)\,h_e(t) - H_e c_h - \rho(t)\,l(t) - z(t) - y, \\
&P(t)\,H_e = H_e \bigl(X_1 - C_{dw}\bigr)V_K(t),\\
&P(t) \ge \mathbf{0}_{q \times q},
\end{align*}
where $y=[y^1,\dots,y^q]^{\top}$, $l(t)=[l^1(t),\dots,l^q(t)]^{\top}$, $z(t)=[z^1(t),\dots,z^q(t)]^{\top}$, and $\rho(t)$ is an upper bound on $\|V_K(t)\|$. Hence, by duality, $\lambda(t)$-contractivity holds if there exist $P(t), V_K(t), \rho(t), \lambda(t)$ satisfying the constraints in~\eqref{LPf}.
\end{proof}

\begin{remark}
The optimizations in~\eqref{lj} and~\eqref{zj} are nonconvex due to the maximization operator. 
This issue can be resolved by expressing $l^j(t)$ and $z^j(t)$ as support–function evaluations of a time-invariant polytope. In particular, recall that $A_{dw}$ and $B_{dw}$ are the stacked 
equality constraints defining the refined disturbance set in~\eqref{dw}. Using these equalities, 
define
\begin{align*}
\mathcal{P}_{dw} = \bigl\{\beta \in \mathbb{R}^{s_c} :
A_{dw}\beta = B_{dw},\, \|\beta\| \le 1\bigr\},
\end{align*}
where $s_c = Ts_w+s_\theta$.

For each facet $j\in\{1,\dots,q\}$, the corresponding direction vectors are
\begin{align*}
a^j_l(t)
&=
M_e(t)\,
\bigl(
H^{j,:}_e G_{dw}^{:,1},\dots,
H^{j,:}_e G_{dw}^{:,s_c}
\bigr)^\top,\\[2pt]
a^j_z(t)
&=
\bigl(
H^{j,:}_e G_{dw}^{:,1} v(t),\dots,
H^{j,:}_e G_{dw}^{:,s_c} v(t)
\bigr)^\top,
\end{align*}
where $ v(t) = D_0^\dagger
\begin{bmatrix}
\bar{x}(t)^\top & \bar{u}(t)^\top
\end{bmatrix}^\top$. The quantities $l^j(t)$ and $z^j(t)$ can then be written as
\begin{align*}
l^j(t)=\sigma_{\mathcal{P}_{dw}}\bigl(a^j_l(t)\bigr)
=\max_{\beta\in\mathcal{P}_{dw}}\,\beta^\top a^j_l(t),\\
z^j(t)=\sigma_{\mathcal{P}_{dw}}\bigl(a^j_z(t)\bigr)
=\max_{\beta\in\mathcal{P}_{dw}}\,\beta^\top a^j_z(t).
\end{align*}

Here, the polytope $\mathcal{P}_{dw}$ is time-invariant, since its geometry depends only on the fixed matrices $A_{dw}$ and $B_{dw}$. Its vertex set can therefore be 
computed offline (e.g., via vertex enumeration or any equivalent polytope representation). In online phase, evaluating $l^j(t)$ and $z^j(t)$ reduces to maximizing finitely many inner products, 
$r^\top a^j_l(t)$ and $r^\top a^j_z(t)$, over the stored vertices $r\in\mathrm{vert}(\mathcal{P}_{dw})$, respectively. This removes the nonconvex maximization and ensures that $l^j(t)$ and $z^j(t)$ can be computed efficiently at every time step.
\end{remark}

\begin{remark}
Under the assumptions of Theorem~\ref{thm-LPf}, suppose that at each time $t$ the optimization problem~\eqref{LPf} is feasible, returns $\lambda(t)\in(0,1)$ and $V_K(t)$, and updates $h_e(t)$ according to~\eqref{hupdate}. By parametrization, the rule $h_e(t{+}1)=\lambda(t)h_e(t)$ with $0<\lambda(t)<1$
implies $\mathcal{E}(t{+}1)=\lambda(t)\,\mathcal{E}(t)\subseteq \mathcal{E}(t),$ so the error tube is nested and monotonically shrinking. The constraints in~\eqref{LPf} further ensure that $\mathcal{E}(t)$ is $\lambda(t)$-contractive in the sense of Definition~\ref{contract}. Hence, the tube parametrization and the contractivity conditions together yield a robust, shrinking error tube. Finally, by jointly optimizing the contraction factor $\lambda(t)$ and the ancillary feedback $V_K(t)$ in~\eqref{LPf}, the TMPC formulation~\eqref{TMPC}
adapts $V_K(t)$ and $\mathcal E(t)$ online, thereby reducing conservatism, enlarging the feasible set of the nominal system, and improving closed-loop performance.
\end{remark} 

\begin{corollary}
\label{lem:CL-dominates-directional}
Consider the tube $\mathcal{E}(t) = \mathcal{P}\bigl(H_e,h_e(t)\bigr)$ in~\eqref{eset}, together with the error dynamics corresponding to (i) the open-loop model obtained from~\eqref{error-dyn2} by allowing the feedback gain to be time-varying and designing $K(t)$ \textup{(}parameterized via $V_K(t)$\textup{)} by~\eqref{LPf}, and (ii) the proposed closed-loop model with adaptive gain in~\eqref{error-zonotope}. For each facet $j\in\{1,\dots,q\}$, using ~\eqref{error-dyn2}, we define 
\[
\Gamma_{{ol}}(t) :=  H_e^{j,:}c_h + F^j + y^j,
\]
where $y^j$ defined in \eqref{yj}, and 
\begin{align}\label{Rj}
F^j 
:= \max_{\substack{x\in\mathcal{X}, u\in\mathcal{U}}}
\bigl| {d^j} [\,x^\top \; u^\top\,]^\top\bigr|,
\end{align}
with
\begin{align}\label{dj}
{d^j} := H_e^{j,:} \begin{bmatrix} \Delta A & \Delta B \end{bmatrix}
\in \mathbb{R}^{1\times(n+m)}.
\end{align}
Furthermore, using ~\eqref{error-zonotope}, we define 
\begin{align*} 
\Gamma^j_{{cl}}(t)
 := H^{j,:}_e c_h 
   +  \rho(t)\,l^j(t)
   + z^j(t) 
   + y^j,
\end{align*}
with $y^j$, $l^j(t)$, and $z^j(t)$ given by~\eqref{yj}--\eqref{zj}.
Suppose that, at time $t$, the facet-wise condition
\begin{equation}\label{eq:facet-condition}
\Gamma^j_{{cl}}(t) < \Gamma^j_{{ol}}(t),
\quad \forall\, j \in \{1,\dots,q\},
\end{equation}
holds. Then, the feasible region of~\eqref{LPf} associated with the proposed closed-loop error dynamics strictly contains the feasible 
region associated with the open-loop error dynamics. In particular, the minimal $\lambda$-contractive tube for the closed-loop error dynamics~\eqref{error-zonotope} is strictly smaller (in the Minkowski sense) than the minimal $\lambda$-contractive tube designed from the open-loop error dynamics~\eqref{error-dyn2}. Moreover, for sufficiently large state and input bounds, or for disturbances and model uncertainty that are not excessively large relative to these bounds, \eqref{eq:facet-condition} holds.

\end{corollary} 
\begin{proof}
For the open-loop error dynamics~\eqref{error-dyn2}, the vector $d^j$ in~\eqref{dj} encodes the model-mismatch direction in facet $j$, so $F^j$ in~\eqref{Rj} represents the worst-case contribution of this mismatch over $\mathcal{X}\times\mathcal{U}$ in direction $H_e^{j,:}$. Consequently, $\Gamma^j_{{ol}}(t)$ upper-bounds the combined effect of open-loop model mismatch and disturbance in facet $j$, and the $\lambda(t)$-contractive tube-update constraint in~\eqref{21b} associated with~\eqref{error-dyn2} can be written, for each $j$, as $P^{j,:}(t)\,h_e(t)
\;\le\; \lambda(t)\,h_e^j(t) - \Gamma^j_{{ol}}(t)$. 

For the closed-loop error dynamics~\eqref{error-zonotope}, the derivation of Theorem~1 shows that, in facet $j$, all contributions of disturbance and model uncertainty, including the closed-loop representation error, are bounded by $\Gamma^j_{{cl}}(t)$, which depends on the current nominal pair $(\bar x(t),\bar u(t))$ rather than on a maximization over $\mathcal{X}\times\mathcal{U}$. If condition~\eqref{eq:facet-condition} holds, then in each facet the closed-loop bound entering the tube-update constraint in~\eqref{21b} is strictly smaller than the corresponding open-loop bound. Since the facetwise contractivity constraints in~\eqref{LPf} are monotone in these bounds, the feasible region of~\eqref{LPf} under the closed-loop error model strictly contains that of the open-loop case, so the optimization can choose a smaller and thus less conservative contraction factor $\lambda(t)$, yielding a tighter $\lambda(t)$-contractive RPI error tube and proving the first claim.


For the second claim, note that \eqref{eq:facet-condition} compares, in
each facet direction $j$: 1) The open-loop bound $\Gamma^j_{{ol}}(t)$, which contains $F^j$,
the support of $\mathcal{X}\times\mathcal{U}$ in direction $d^j$, and
thus grows as $\mathcal{X}$ and $\mathcal{U}$ are enlarged along
directions with $d^j\neq 0$. 2) The closed-loop bound $\Gamma^j_{{cl}}(t)$, which depends only on $(\bar x(t),\bar u(t))$,
the refined sets $\mathcal{M}^c_{{ol}}$ and
$\mathcal{M}^c_{{cl}}(t)$, and the tube size via $h_e(t)$ and
$M_e(t)$, and does not involve a maximization over
$\mathcal{X}\times\mathcal{U}$. In practice,
$\Gamma^j_{{cl}}(t)$ remains moderate because (i) the nominal
trajectory stays in the interior of $\mathcal{X}\times\mathcal{U}$, (ii) the closed-loop model set $\mathcal{M}^c_{{cl}}(t)$ shrinks as more data are collected and $V_K(t)$ is updated, and (iii) the tube size $h_e(t)$ and the bound
$M_e(t)$ decrease as the contractive tube is updated, further reducing
$\rho(t)\,l^j(t)$ and $z^j(t)$. Thus, for sufficiently large but physically meaningful bounds
$\mathcal{X},\mathcal{U}$ along the mismatch directions, and for
disturbances and model uncertainty of comparable size, it is reasonable
to expect that condition~\eqref{eq:facet-condition} holds.
\end{proof}

\subsection{Recursive Feasibility $\&$ Stability Analysis}
This subsection analyzes the recursive feasibility and closed-loop stability of the presented
elastic TMPC scheme. In particular, we show that the alternating structure between the TMPC 
problem and the tube–gain update ensures that: (i) the error tube remains $\lambda(t)$–contractive at 
every step, (ii) the tightened constraints remain feasible at all future times, and (iii) the nominal closed loop admits a Lyapunov decrease, and (iv) a polyhedral Lyapunov candidate for the error tube guarantees exponential stability of the closed-loop error dynamics. These results formalize how the learned gain $V_K(t)$, the refined model sets, and the shrinking tube interact to guarantee robust constraint satisfaction and stable closed-loop behavior.

\begin{theorem}\label{recfeas}
Consider the uncertain system~\eqref{system} with error dynamics~\eqref{error-zonotope}
under Assumptions~1--8, and let the error tube $\mathcal{E}(t)$ be defined as
in~\eqref{eset}. At each time $t$, the controller operates in two steps:
\begin{enumerate}[label=(\roman*)]
  \item Solve the tube–gain update problem~\eqref{LPf} to obtain
        $\mathcal{E}(t)$ and $V_K(t)$.
  \item Solve the TMPC problem~\eqref{TMPC} with the updated tube
        $\mathcal{E}(t)$ and gain $V_K(t)$, and apply to the plant the control input \eqref{optu}.
\end{enumerate}
Suppose that the tube–gain update problem~\eqref{LPf} is feasible for all
$t \ge t_0$, and that the TMPC problem~\eqref{TMPC} is feasible at some
time $t_0$. Then the TMPC problem~\eqref{TMPC} remains feasible for all
subsequent times $t \ge t_0$; in particular, feasibility at time $t$ implies
feasibility at time $t{+}1$.
\end{theorem}

\begin{proof}
Fix an arbitrary $t \ge t_0$. By assumption, the tube--gain update problem \eqref{LPf} is feasible at time $t$. Together with Assumptions~1--8, Theorem~\ref{thm-LPf} implies that the updated error tube $\mathcal{E}(t)$ is $\lambda(t)$-contractive and RPI for the error
dynamics, and that the next tube satisfies $\mathcal{E}(t{+}1) \subseteq \mathcal{E}(t)$. By feasibility of \eqref{TMPC} at time $t$, there exist nominal trajectories $\{\bar{x}^\star(t+k\mid t)\}_{k=0}^{N-1}$ and $\{\bar{u}^\star(t+k\mid t)\}_{k=0}^{N-1}$ satisfying the tightened state and input constraints and the terminal constraint $\bar{x}^\star(t{+}N\mid t) \in \mathcal{T}$.

To prove feasibility of \eqref{TMPC} at time $t{+}1$, we construct the standard
shifted candidate:
\begin{align*}
\tilde{x}(t{+}k\mid t{+}1) &:= \bar{x}^\star(t{+}1{+}k\mid t),\quad\,  k = 0,\dots,N,\\
\tilde{u}(t{+}k\mid t{+}1) &:= \bar{u}^\star(t{+}1{+}k\mid t),\quad\,  k = 0,\dots,N{-}2,\\
\tilde{u}(t{+}N{-}1\mid t{+}1) &:= K_{\mathcal{T}}\big(\bar{x}^\star(t{+}N\mid t)\big),
\end{align*}
where $K_\mathcal{T}(\cdot)$ satisfies the tightened input constraints by Assumption~\ref{terminal}.

For the tightened state constraints, feasibility at time $t$ implies
$\bar{x}^\star(t{+}k\mid t) \oplus \mathcal{E}(t) \subseteq \mathcal{X}$ for
$k=0,\dots,N-1$. Since $\mathcal{E}(t{+}1)\subseteq \mathcal{E}(t)$ and the
Minkowski sum is monotone, for $k=0,\dots,N{-}1$ we obtain
\[
\tilde{x}(t{+}k\mid t{+}1)\oplus \mathcal{E}(t{+}1)
\subseteq \bar{x}^\star(t{+}1{+}k\mid t)\oplus \mathcal{E}(t)
\subseteq \mathcal{X}.
\]

Moreover, feasibility at time $t$ gives $\bar{x}^\star(t{+}N\mid t)\in\mathcal{T}$.
By Assumption~\ref{terminal}, $\mathcal{T}$ is positively invariant under
$K_\mathcal{T}(\cdot)$, so
$\tilde{x}(t{+}N\mid t{+}1)\in\mathcal{T}$,
and all tightened state constraints at time $t{+}1$ are satisfied.

For the tightened input constraints, feasibility at time $t$ implies
$\bar{u}^\star(t{+}k\mid t) \oplus U_0 V_K(t)\,\mathcal{E}(t) \subseteq \mathcal{U}$
for all $k = 0,\dots,N{-}1$. By Theorem~\ref{thm-LPf}, the updated tube--gain pair satisfies
$V_K(t{+}1)\,\mathcal{E}(t{+}1) \subseteq V_K(t)\,\mathcal{E}(t)$, so for
$k=0,\dots,N{-}2$ we obtain
\begin{align*}
\tilde{u}(t{+}k\mid t{+}1)\oplus U_0 V_K(t{+}1)\,\mathcal{E}(t{+}1)
&\subseteq\\ \bar{u}^\star(t{+}1{+}k\mid t) \oplus U_0 V_K(t)\,\mathcal{E}(t)
\subseteq \mathcal{U},
\end{align*}
which shows that all tightened input constraints are satisfied at time $t{+}1$. For $k=N{-}1$, we have
$\tilde{u}(t{+}N{-}1\mid t{+}1)=K_\mathcal{T}(\bar{x}^\star(t{+}N\mid t))$, and, by
Assumption~\ref{terminal}, the corresponding input constraint also holds.

Therefore, the shifted pair
$\{\tilde{x}(t{+}k\mid t{+}1),\tilde{u}(t{+}k\mid t{+}1)\}_{k=0}^{N-1}$
is feasible for \eqref{TMPC} at time $t{+}1$. Since $t\ge t_0$ was arbitrary
and \eqref{LPf} is feasible for all $t \ge t_0$, the same argument can be
repeated at each step, proving that feasibility of \eqref{TMPC} at $t_0$
propagates to all $t \ge t_0$.
\end{proof}

\begin{remark}
Consider the TMPC problem~\eqref{TMPC} at the initial time $t_0$ with a conservative
fixed-tube pair $(\hat{\mathcal E}(t_0),\hat V_K(t_0))$ inducing the tightened sets
$\mathcal X \ominus \hat{\mathcal E}(t_0)$ and
$\mathcal U \ominus U_0 \hat V_K(t_0)\hat{\mathcal E}(t_0)$.
If the tube--gain update problem~\eqref{LPf} is feasible at $t_0$, 
Theorem~\ref{thm-LPf} yields a less conservative pair $({\mathcal E}(t_0),V_K(t_0))$, hence $\mathcal X \ominus \hat{\mathcal E}(t_0) \subseteq \mathcal X \ominus {\mathcal E}(t_0),$ and $\mathcal U \ominus U_0 \hat V_K(t_0)\hat{\mathcal E}(t_0) 
\subseteq \mathcal U \ominus U_0 V_K(t_0){\mathcal E}(t_0)$. Thus, the adaptive tube--gain update enlarges the initial nominal feasible region of~\eqref{TMPC} and makes the initial feasibility requirement of recursive feasibility less restrictive and more likely under larger disturbance levels than in fixed-tube designs.
\end{remark}


\begin{proposition}
Let $J(\bar x^\star(t|t))$ denote the optimal cost of \eqref{TMPC} at time $t$. Under
Theorem~\ref{recfeas} and Assumption~\ref{terminal}, the optimal cost satisfies
\[
J(\bar x^\star(t|t+1))\;\le\; J(\bar x^\star(t|t))\;-\;\mathcal{L}\big(\bar x^\star(t|t),\bar u^\star(t|t)\big),
\]
where $\bigl(\bar x^\star(t\mid t),\bar u^\star(t\mid t)\bigr)$ is the first
state–input pair of the optimal nominal state sequence
$[\bar x^\star(t\mid t),\dots,\bar x^\star(t+N\mid t)]$
and input sequence
$[\bar u^\star(t\mid t),\dots,\bar u^\star(t+N-1\mid t)]$
returned by~\eqref{TMPC} at time $t$, respectively. Consequently,
$J$ is a Lyapunov function for the nominal closed loop and the origin is
asymptotically stable.
\end{proposition}

\begin{proof}
The proof follows the approach in \cite{stable,MPCbook1}.
\end{proof}

\begin{theorem}
\label{tube-lyap}
Let $V:\mathcal{E}(t)\rightarrow \mathbb{R}$ be defined by
\begin{equation}\label{lyap}
V\bigl(e(t)\bigr)
:= \max_{j \in \{1,\dots,q\}}\frac{H_e^{j,:} e(t)}{h_e^j(t)}.
\end{equation}
Suppose that, over a finite horizon of length $\mathcal N$, the optimization problem~\eqref{LPf} is feasible at each time step and returns a sequence of decision variables $\mathcal{V}_K = \{ V_K(0),\dots,V_K(\mathcal N-1) \}$ with corresponding contraction factors $\Lambda = \{\lambda(0),\dots,\lambda(\mathcal N-1)\},$ with $\lambda(t)\in(0,1),$ and define the worst-case contraction factor as
\begin{equation} \label{lambdabar}
\bar\lambda \;:=\; \max_{t\in\{0,\dots,\mathcal N-1\}} \lambda(t).
\end{equation}
Then $V\bigl(e(t)\bigr)$ is positive definite and serves as a polytopic Lyapunov candidate for the error tube
$\mathcal{E}(t)$ defined in~\eqref{eset}. Moreover, the following properties hold over the horizon $t = 0,\dots,\mathcal N-1$:

\begin{enumerate}[label=(\roman*)]
\item Disturbance-free case. If $w(t)\equiv 0$, then along any
switching sequence generated by~\eqref{LPf} one has
\[
V\bigl(e(t+1)\bigr)
\;\le\;
\bar \lambda\,V\bigl(e(t)\bigr).
\]
Thus, the error decays exponentially over the horizon.
\item Disturbance case. In the presence of bounded disturbances and model uncertainty, there exists a constant $c_{\bar w}>0$ such that
\[
V\bigl(e(t+1)\bigr)
\;\le\;
\bar\lambda\,V\bigl(e(t)\bigr) + c_{\bar w},
\]
Consequently, the error trajectories converge exponentially fast to a bounded positively invariant neighborhood of the origin.
\end{enumerate}
\end{theorem}

\begin{proof}
We first verify that $V(e(t))$ defines a polytopic Lyapunov function on the
error tube $\mathcal{E}(t)$ in~\eqref{eset}, i.e., $V$ is positive definite
on $\mathcal{E}(t)$. Since $0 \in \mathcal{E}(t)$, we have $H_e^{j,:} 0 \le h_e^j(t)$ for all
$j = 1,\dots,q$, and $0 \in \operatorname{int}\mathcal{E}(t)$ implies
$h_e^j(t) > 0$ for all $j$. Clearly,
\[
V(0) = \max_{j} \frac{H_e^{j,:}\, 0}{h_e^j(t)} = 0.
\]

Let $e(t) \in \mathcal{E}(t)\setminus\{0\}$ be arbitrary and consider the ray $r(\bar\alpha) := \bar\alpha e(t),$ with $\bar\alpha \ge 0$. Since $0 \in \operatorname{int}\mathcal{E}(t)$ and $\mathcal{E}(t)$ is compact, there exists a finite maximal scalar $\bar\alpha^\star > 0$ such that $\bar\alpha^\star e(t) \in \partial\mathcal{E}(t)$, and $\bar\alpha e(t) \in \mathcal{E}(t)$ for all $\bar\alpha \in [0,\bar\alpha^\star]$. Since $e(t)\neq 0$ and $0$ is an interior point, we must have $\bar\alpha^\star \ge 1$. At the boundary point $\bar\alpha^\star e(t)$, at least one constraint is active, so there exists $j \in \{1,\dots,q\}$ with
$H_e^{j,:}(\bar\alpha^\star e(t)) = h_e^j(t)$. Hence,
\[
H_e^{j,:} e(t) = \frac{h_e^j(t)}{\bar\alpha^\star} > 0
\quad\Rightarrow\quad
\frac{H_e^{j,:} e(t)}{h_e^j(t)} = \frac{1}{\bar\alpha^\star} > 0.
\]
Therefore,
\[
V\bigl(e(t)\bigr)
= \max_{j} \frac{H_e^{j,:} e(t)}{h_e^j(t)}
\;\ge\; \frac{H_e^{j,:} e(t)}{h_e^j(t)}
= \frac{1}{\bar\alpha^\star}
> 0.
\]

Thus, $V\bigl(e(t)\bigr) \ge 0$ for all $e(t)\in\mathcal{E}(t)$, and $V\bigl(e(t)\bigr) = 0$ if and only if $e(t)=0$. Hence, $V\bigl(e(t)\bigr)$ is positive definite on $\mathcal{E}(t)$ and serves as a polytopic Lyapunov candidate. We now establish its decay under $\lambda(t)$–contractivity in two cases: (i) the disturbance–free case and (ii) the case with disturbances.

\emph{(i) Disturbance–free case.}
Since $e(t)\in\mathcal{E}(t)$, the definition of $\mathcal{E}(t)$ implies $H_e^{j,:} e(t) \le h_e^j(t)$ for all $j\in\{1,\dots,q\}$. As $h_e^j(t) > 0$ for all $j$, we obtain
\[
\frac{H_e^{j,:} e(t)}{h_e^j(t)} \le 1, \qquad j=1,\dots,q,
\]
hence $0 \le V\bigl(e(t)\bigr) \le 1.$ Thus, for any $e(t)\in\mathcal{E}(t)$ we can write $V\bigl(e(t)\bigr)= c$ for some $c\in[0,1]$. Moreover, by the definition of $V$, the condition
\[
\frac{H_e^{j,:} e(t)}{h_e^j(t)} \le c \quad \forall j
\;\Longleftrightarrow\;
H_e e(t) \le c\, h_e(t),
\]
implies 
\begin{equation}\label{Ec}
e(t) \in c\,\mathcal{E}(t) =: \mathcal{E}_c(t)
= \mathcal{P}\bigl(H_e, c\,h_e(t)\bigr).
\end{equation}

Solving the optimization problem~\eqref{LPf} over the finite horizon
$\mathcal N$ yields a sequence of tube–gain decisions $\mathcal{V}_K$
with associated contraction factors $\Lambda$. According to
Theorem~\ref{thm-LPf}, in the disturbance-free case (i.e., $w(t)\equiv 0$),
the error dynamics $e(t+1) = A^{K(t)} e(t)$, with
$A^{K(t)} := A + BU_0 V_K(t)$ and $V_K(t)\in\mathcal{V}_K$, satisfy
$A^{K(t)} \mathcal{E}(t) \subseteq \lambda(t)\,\mathcal{E}(t)$ for all $t=0,\dots,\mathcal N-1$, where
$\lambda(t) \in (0,1)$. This implies $e(t+1) \in \lambda(t)\,\mathcal{E}(t)= \mathcal{E}(t+1)$. 

Using $e(t) \in \mathcal{E}_c(t)$ from~\eqref{Ec} with $c := V\bigl(e(t)\bigr)\in[0,1]$, the same $\lambda(t)$–contractivity property from Theorem~\ref{thm-LPf}
applied to $\mathcal{E}_c(t)$ yields
$e(t+1) \in \lambda(t)\,\mathcal{E}_c(t)$. Hence
\[
H_e^{j,:} e(t+1) \le \lambda(t)\,c\,h_e^j(t)
\quad\Rightarrow\quad
\frac{H_e^{j,:} e(t+1)}{h_e^j(t)} \le \lambda(t)\,c,
\]
for $j=1,\dots,q.$ Taking the maximum over $j$ gives
\[
\max_j \frac{H_e^{j,:} e(t+1)}{h_e^j(t)}
\;\le\; \lambda(t)\,c,
\]
and thus, $V\bigl(e(t+1)\bigr) \le \lambda(t)\,V\bigl(e(t)\bigr)\le \bar\lambda\,V\bigl(e(t)\bigr),$ with $\bar\lambda$, defined in \eqref{lambdabar}. By induction this yields
\begin{align*}
V\bigl(e(t)\bigr)
&\;\le\;
\Biggl(\prod_{k=0}^{t-1} \lambda(k)\Biggr)
V\bigl(e(0)\bigr)\;\le\;
\bar\lambda^{\,t}\,V\bigl(e(0)\bigr).
\end{align*}

Since $\bar\lambda \in (0,1)$, $V(e(t))$ decays to zero exponentially fast, which completes the proof of case~(i).

\emph{(ii) Disturbance case.} Consider the uncertain error dynamics~\eqref{error-zonotope}
and define the lumped uncertainty $\bar w(t)$, which collects all disturbances and model
uncertainties and is confined to the zonotope $\langle G_{\bar w}, c_{\bar w} \rangle$.
Using positive homogeneity and subadditivity of $V$, one has
\[
V\bigl(e(t+1)\bigr)
\le V\bigl(A^{K(t)} e(t)\bigr) + V\bigl(\bar w(t)\bigr).
\]

Define
\[
c_{\bar w} := \max_{\|\zeta_{\bar w}\|\le 1} V\bigl( c_{\bar w} + G_{\bar w} \zeta_{\bar w}\bigr),
\]
which is finite since $\bar w$ is bounded. Hence
\[
V\bigl(e(t+1)\bigr) \le \bar\lambda\,V\bigl(e(t)\bigr) + c_{\bar w},
\]
which implies that $V\bigl(e(t)\bigr)$ converges exponentially to a bounded, positively invariant neighborhood of the origin. This completes the proof of case~(ii).

\end{proof}

\section{Simulation Results}

To validate the proposed elastic TMPC approach in practical scenarios, we consider the kinematic model of the Mecanum-drive ROSbot~XL platform (shown in Fig.~\ref{a}), where $x = [\,x_r\;\;y_r\;\;\varphi_r\,]^\top$ represents the planar position and yaw, and $u = [\,\omega^1\;\omega^2\;\omega^3\;\omega^4\,]^\top$ denotes the wheel angular rates. The body–velocity relation is
\begin{equation*}
\begin{bmatrix}
\dot x_r\\\dot y_r\\\dot\varphi_r
\end{bmatrix}
=
\frac{R_w}{4\ell_{ab}}
\begin{bmatrix}
\ell_{ab} &  \ell_{ab} &  \ell_{ab} &  \ell_{ab}\\
\ell_{ab} & -\ell_{ab} &  \ell_{ab} & -\ell_{ab}\\
1         & -1         & -1         &  1
\end{bmatrix}
\!
\begin{bmatrix}
\omega^1\\ \omega^2\\ \omega^3\\ \omega^4
\end{bmatrix},
\label{eq:rosbot-ct}
\end{equation*}
with wheel radius $R_w = 50\,\mathrm{mm}$ and $\ell_{ab} = \ell_a+\ell_b$, where $\ell_a = 134.84\,\mathrm{mm}$ and $\ell_b = 85\,\mathrm{mm}$ denote the longitudinal and lateral half-spacings shown in Fig.~\ref{b} \cite{rosbot}. For sampling time $t_s$, the discrete-time kinematic model is obtained as 
\begin{equation*} 
A^\star=\mathbf{I}_3,\quad
B^\star=\frac{R_w\,t_s}{4\ell_{ab}}
\begin{bmatrix}
\ell_{ab} &  \ell_{ab} &  \ell_{ab} &  \ell_{ab} \\
\ell_{ab} & -\ell_{ab} &  \ell_{ab} & -\ell_{ab} \\
1         & -1         & -1         &  1
\end{bmatrix}.
\label{eq:rosbot-dt}
\end{equation*}

The additive disturbance $w(t)$ lies in a zonotope with $c_h = 0$ and
\[
G_h = \alpha
\begin{bmatrix}
    0.05 & 0.08 \\
    0.01 & 0.06\\
    0.03 & -0.01
\end{bmatrix},
\]
where $\alpha$ is a noise scaling factor. The admissible state set is a box
constraint $\mathcal{X} = \big\{ x \in \mathbb{R}^3 : |x_r| \le 4,\; |y_r| \le 4,\;
|\varphi_r| \le \pi/2 \big\},$ and the input set $\mathcal{U}$ restricts the wheel speeds to given bounds $|\omega^i| \le 100, \,i = 1,\dots,4.$ Prior knowledge on the dynamics is encoded as a constrained matrix zonotope $\langle G_\theta, C_\theta, \bar A_w, \bar B_w\rangle$, where $C_\theta$ and
$G_\theta$ are obtained from 17 offline data points generated by applying a
stabilizing input under a different disturbance zonotope
\[
\mathcal{M}_{wp} = \big\langle
\begin{bmatrix}
    0.03 & -0.01 \\
    -0.04 & 0.05 \\
    -0.01 & 0
\end{bmatrix},
\begin{bmatrix}
    1 \\ -1 \\ 0
\end{bmatrix},
\bar A_w,\bar B_w
\big\rangle,
\]
with $\bar A_w$ and $\bar B_w$ computed as in~\eqref{Aw}. In all simulations, the prediction horizon $N$ is set to 10, $Q = 20\,\mathbf{I}_3$, $R = 0.1\,\mathbf{I}_4$, and $\sigma = 1$.

\begin{figure}[t]
  \centering
  \begin{subfigure}{0.49\linewidth}
    \includegraphics[width=\linewidth, keepaspectratio]{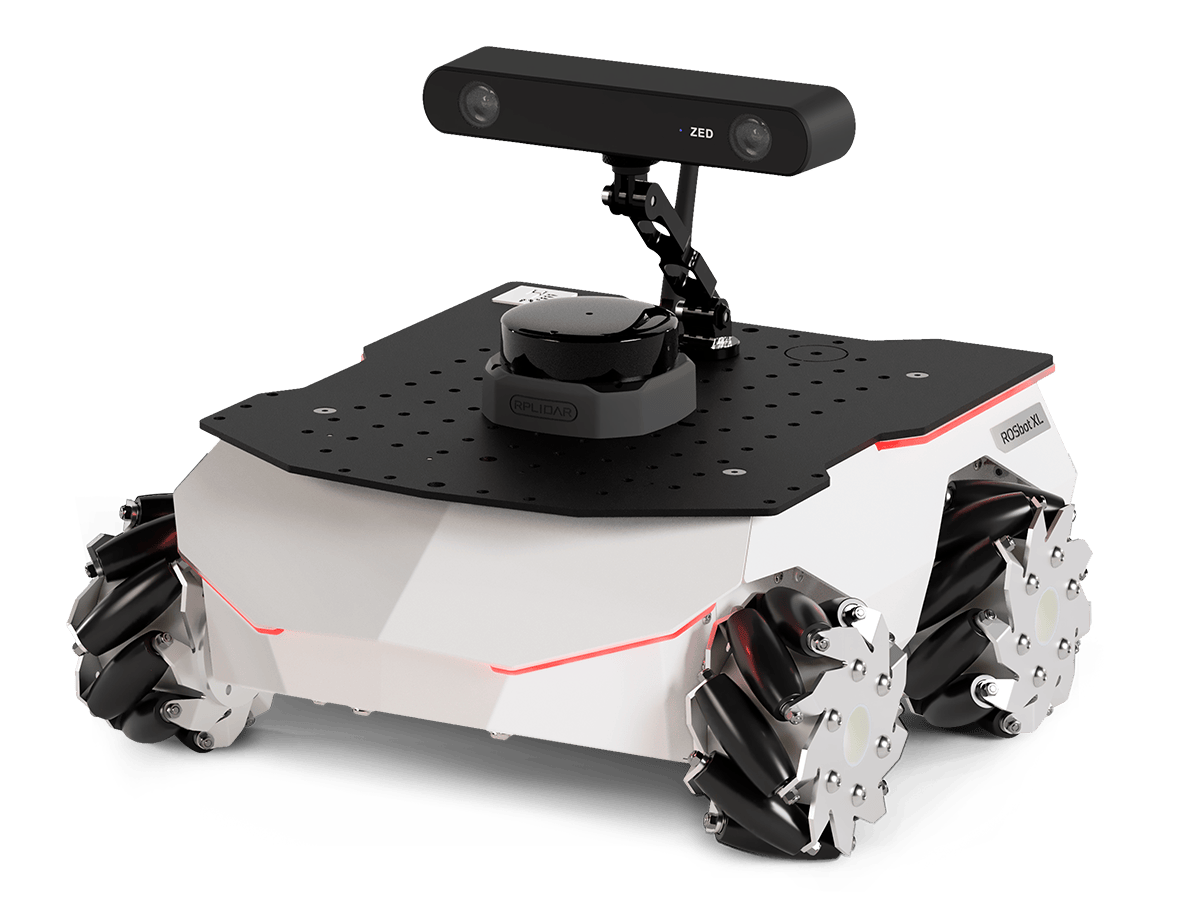}
    \caption{}
    \label{a}
  \end{subfigure}%
  \hspace{0.01\linewidth}
  \begin{subfigure}{0.49\linewidth}
    \includegraphics[width=\linewidth, height=1.5in]{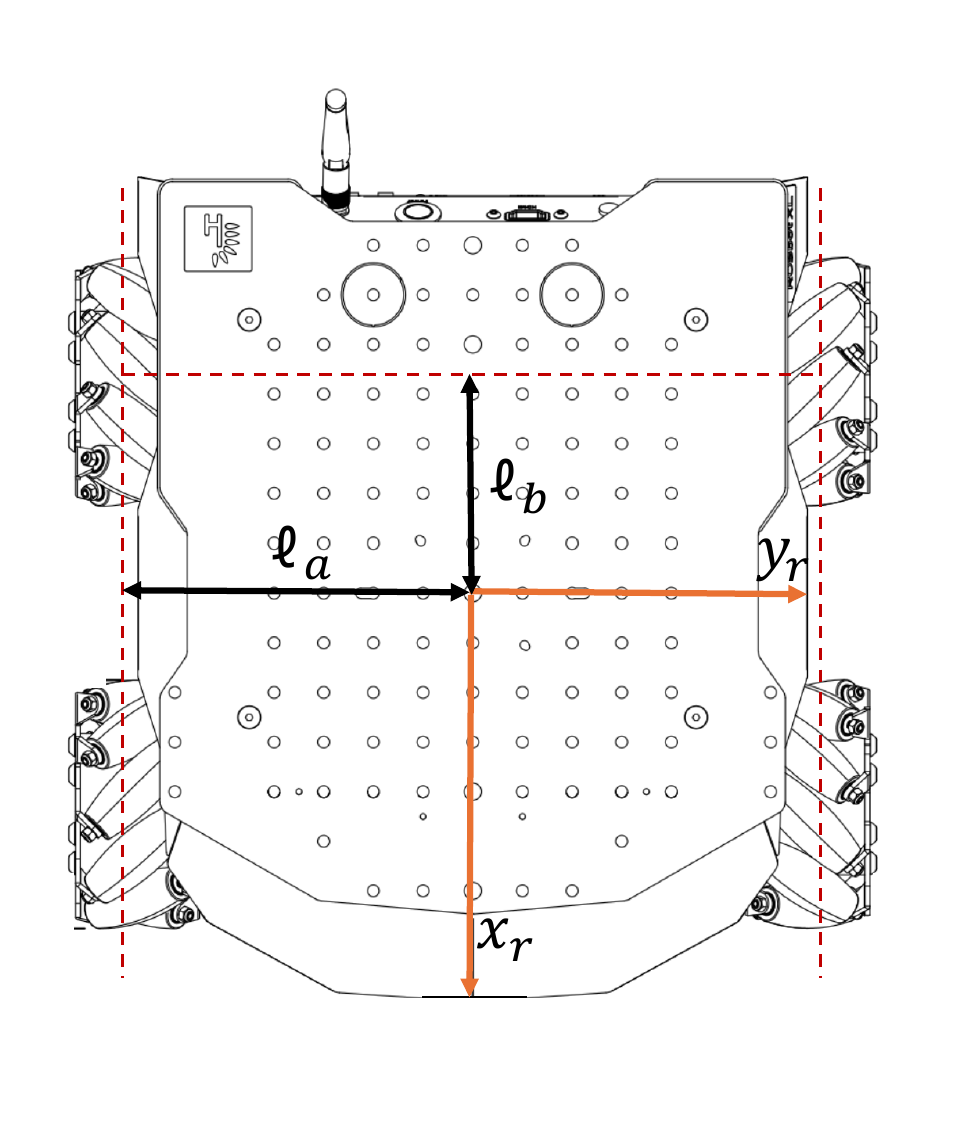}
    \caption{}
    \label{b}
  \end{subfigure}
  \caption{\small{ROSbot XL equipped with Mecanum wheels.}}
  \label{rosbot}
\end{figure}
To illustrate how adaptive gain updates and tube shrinking reduce conservatism and improve feasibility in the TMPC problems, we compare feasibility rates across different noise levels and data lengths. Fig.~\ref{alpha} shows the percentage of feasibility versus the noise scaling factor $\alpha$, while Fig.~\ref{TT} demonstrates feasibility as a function of the data length $T$. In both plots, we benchmark the proposed elastic TMPC against the data-driven tube-based zonotopic predictive control (TZPC) scheme of~\cite{un2}, which is based on the error dynamics~\eqref{error-dyn2} with a fixed feedback gain $K$ and a fixed robust invariant tube. For our method, we consider two variants: (i) a data-only refinement, where the disturbance tube is updated solely from data-consistency constraints defined in \eqref{Aw}; and (ii) a data--prior refinement as in~\eqref{dw}, where the disturbance model set is further tightened by intersecting the data-consistent set with the physical prior.

\begin{figure}[t]
  \centering
  \begin{subfigure}[t]{0.85\linewidth}
    \centering
    \includegraphics[width=\linewidth,keepaspectratio]{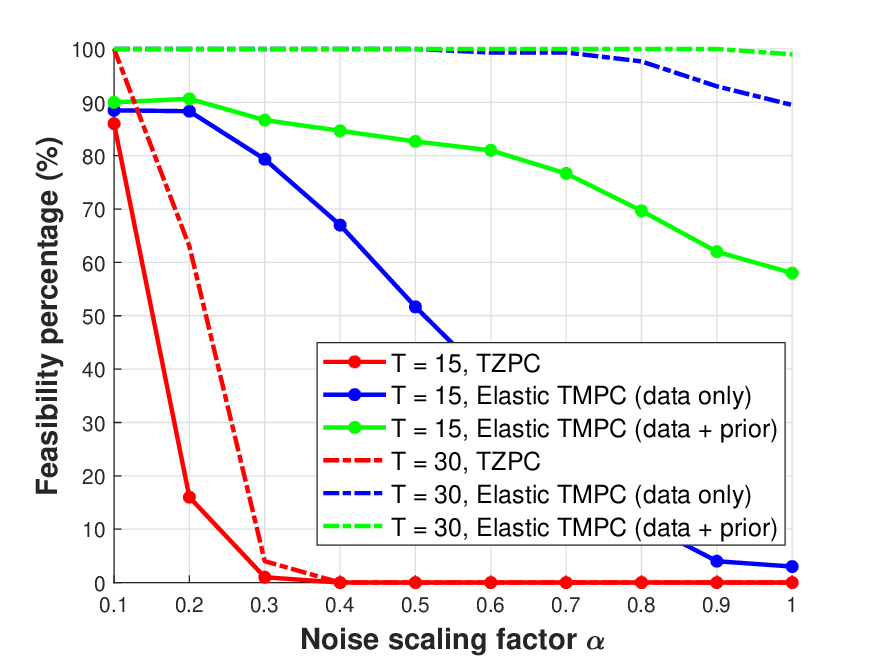}
    \subcaption{}
    \label{alpha}
  \end{subfigure}
  \begin{subfigure}[t]{0.85\linewidth}
    \centering
    \includegraphics[width=\linewidth,keepaspectratio]{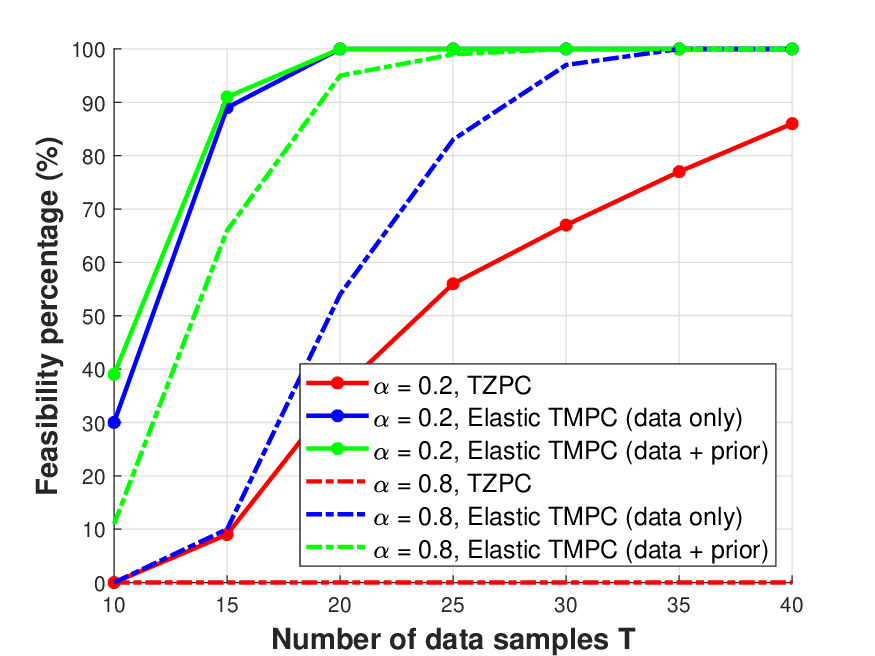}
    \subcaption{}
    \label{TT}
  \end{subfigure}
  \caption{\small{Feasibility percentage versus (a) noise scaling factor $\alpha$ and (b) data length $T$, each averaged over 100 Monte Carlo simulations.}}
  \label{alpha-TT}
\end{figure}

In Fig.~\ref{alpha}, we plot the feasibility percentage as a function of the noise scaling factor $\alpha$ for two data lengths, $T=15$ and $T=30$. For $T=15$, both variants of the proposed adaptive TMPC substantially outperform TZPC: the data-only refinement (blue) decreases from about $90\%$ to values close to $0\%$ as $\alpha$ grows from $0.1$ to $1$, while the data--prior refinement (green) remains markedly more robust, dropping from about $90\%$ to roughly $60\%$. Over the same range, TZPC (red) starts slightly under $90\%$ at $\alpha=0.1$ but its feasibility collapses to nearly zero already around $\alpha= 0.4$. For $T=30$, the advantage of the proposed approach becomes even more pronounced: the data--prior refinement (green) maintains $100\%$ feasibility across all noise levels, and the data-only refinement (blue) stays above about $90\%$ even at $\alpha=1$, whereas TZPC again rapidly loses feasibility as the noise level increases. This weak noise tolerance of TZPC is a direct consequence of its conservative design, which employs a fixed feedback gain and a fixed robust invariant tube derived from an error description based solely on the open-loop model set defined in \eqref{error-dyn2}. As a result, the tube size scales with the admissible sets $\mathcal{X}$ and $\mathcal{U}$, leading to overly large tubes, increased conservatism, and reduced noise tolerance.

In Fig.~\ref{TT}, we plot the feasibility percentage versus the number of data samples $T$ for two noise scaling factors, $\alpha=0.2$ and $\alpha=0.8$. For $\alpha=0.2$, the feasibility percentage of all three methods improve with more data but at very different rates: TZPC (red) increases only gradually, from about $0\%$ at $T=10$ to roughly $85\%$ at $T=40$, whereas the elastic TMPC schemes increase much more rapidly—the data-only variant (blue) rises from about $30\%$ at $T=10$ to $100\%$ at $T=20$, and the data--prior variant (green) from about $40\%$ to $100\%$ over the same range. For the higher noise level $\alpha=0.8$, the contrast is sharper: TZPC remains at $0\%$ feasibility for all $T\in[10,40]$, indicating that additional data alone cannot overcome its fixed conservative tube. In contrast, the proposed elastic TMPC still exploits the data effectively: the data-only refinement climbs from $0\%$ at $T=10$ to $100\%$ by $T=35$, while the data--prior refinement starts around $10\%$ at $T=10$ and attains $100\%$ feasibility around $T=25$, already exceeding $90\%$ at $T=20$. Consequently, even at the higher noise level, the proposed approach can achieve $100\%$ feasibility with only a minimal amount of online noisy data. This behavior is consistent with the fact that, as $T$ increases, the realized disturbance set becomes tighter. The elastic TMPC schemes exploit these refined sets—combining data-consistency constraints and priors—and use them to update both the tube and the feedback gain online, which in turn yields higher feasibility. By contrast, TZPC keeps a conservative fixed tube and gain tied to $\mathcal X$ and $\mathcal U$ without refining the realized disturbance set, so for $\alpha = 0.8$ its feasibility does not improve over the explored data range and would require substantially more data to achieve comparable performance.

We next fix the data length at $T=20$ and the noise--scaling parameter at $\alpha=0.7$, and consider phase portraits with target state $[-3.5,,-3.5,,-\pi/4]^\top$. Fig.~\ref{TMPCplot} shows a representative closed-loop realization in which the system starts from the boundary of the admissible set at $x(0)=[4,\,4,\,\pi/2]^\top$ and is driven to the target under the proposed TMPC approach. To quantify the benefit of fusing physical priors with a small batch of online noisy data, we compare two refinement strategies for the realized disturbance set used in the tube--gain update. In Fig.~\ref{dd}, only the data--consistency constraint~\eqref{dcons} is enforced. This produces the data--refined set $\mathcal{M}_w$ with constraint matrices defined in \eqref{Aw}, which is then used in the tube--gain update optimization~\eqref{LPf}. In Fig.~\ref{pp}, we additionally impose prior consistency by intersecting the data--consistent set with the physical prior, yielding the tighter set $\mathcal{M}_{dw}$ in~\eqref{dw}, again employed in~\eqref{LPf}. As shown in Fig.~\ref{TMPCplot}, both cases lead to tube contraction over time as more data are incorporated. Nevertheless, enforcing the physical prior eliminates disturbance models incompatible with the priors, yielding a much smaller initial uncertainty set. This reduction leads to a tighter error tube, which in turn reduces conservatism in the nominal planning while preserving safety.


\begin{figure}[ht]
  \centering
  \begin{subfigure}[t]{0.85\linewidth}
    \centering
    \includegraphics[width=\linewidth,keepaspectratio]{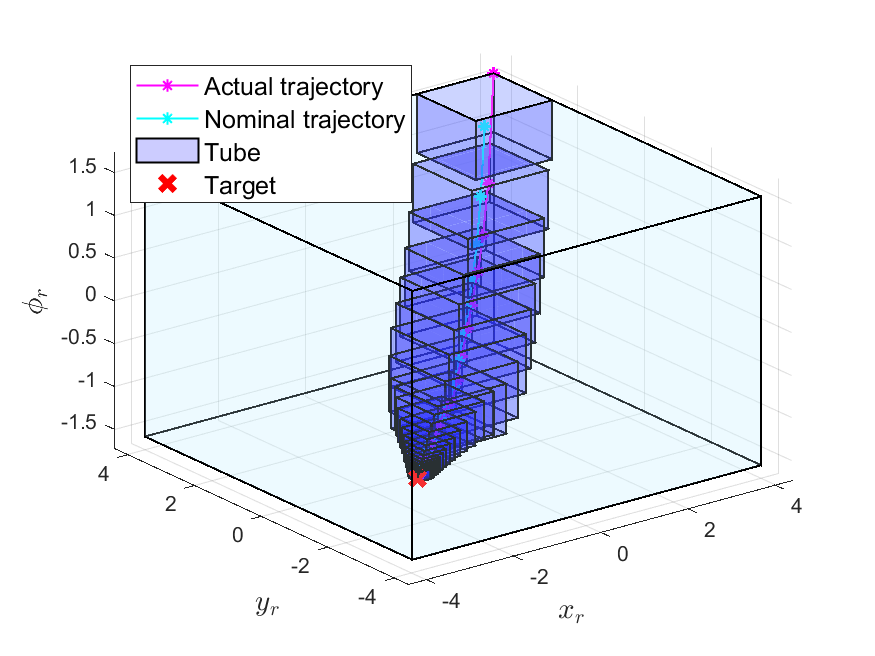}
    \subcaption{}
    \label{dd}
  \end{subfigure}
  \begin{subfigure}[t]{0.85\linewidth}
    \centering
    \includegraphics[width=\linewidth,keepaspectratio]{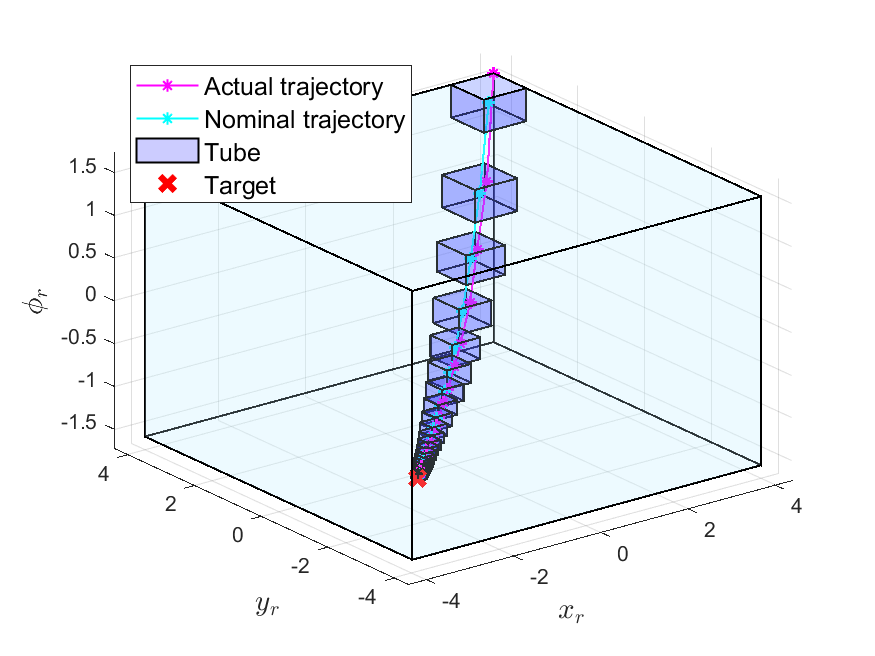}
    \subcaption{}
    \label{pp}
  \end{subfigure}
  \caption{\small{Phase portrait of the mobile robot's trajectory in the state space, starting from $[4,\,4,\,\pi/2]^\top$ and steering to the target $[-3.5,\,-3.5,\,-\pi/4]^\top$. (a) Data-only refinement. (b) Data--prior refinement.}}
  \label{TMPCplot}
\end{figure}
\section{Conclusion}
A novel elastic tube-based MPC framework is presented for unknown discrete-time linear systems with bounded disturbances that fuses data with prior physical knowledge. A zonotopic disturbance set is initialized and refined into a constrained matrix zonotope, yielding constrained matrix-zonotope model sets for both open- and closed-loop dynamics. By separating open-loop mismatch from closed-loop effects and adaptively co-designing the tube and ancillary feedback, the controller enforces contractive zonotopic tubes that ensure robust positive invariance, enlarge feasibility margins, and improve disturbance tolerance. Recursive feasibility is established and exponential stability is certified via a polyhedral Lyapunov function. Simulations demonstrate higher feasibility and robust constraint satisfaction using only a small batch of online data; future work will extend the framework to nonlinear stochastic systems.

\bibliographystyle{unsrt} 
\bibliography{ref.bib}



\end{document}